\documentclass[sigconf,nonacm]{aamas} 

\usepackage{balance} % for balancing columns on the final page
\RequirePackage{amsmath}
\usepackage{amsfonts}
\usepackage{framed}
\usepackage{graphicx}
\usepackage{float}
\usepackage{algorithm}
\usepackage[noend]{algorithmic}
\usepackage{multirow}
\usepackage{multicol}
\usepackage{color}
\usepackage{url}
\usepackage{geometry}
\usepackage{enumerate}
\usepackage{tikz}

\newcommand{\shuf}{\mathtt{shuf}}
\newcommand{\sfsshuf}{\mathtt{sfs}\text{-}\mathtt{shuf}}
\newcommand{\gsfsshuf}{\mathtt{gsfs}\text{-}\mathtt{shuf}}
\newcommand{\MC}{\mathtt{MC}}

\setcopyright{ifaamas}
\acmDOI{}
\acmPrice{}
\acmISBN{}

\title{Incentives for Early Arrival in Cost Sharing}

\author{Junyu Zhang}
\affiliation{
  \institution{Key Laboratory of Intelligent Perception and Human-Machine Collaboration, ShanghaiTech University}
  \city{Shanghai}
  \country{China}}
\email{zhangjy22022@shanghaitech.edu.cn}

\author{Yao Zhang}
\affiliation{
  \institution{Key Laboratory of Intelligent Perception and Human-Machine Collaboration, ShanghaiTech University}
  \city{Shanghai}
  \country{China}}
\email{zhangyao1@shanghaitech.edu.cn}

\author{Yaoxin Ge}
\affiliation{
  \institution{Key Laboratory of Intelligent Perception and Human-Machine Collaboration, ShanghaiTech University}
  \city{Shanghai}
  \country{China}}
\email{geyx@shanghaitech.edu.cn}

\author{Dengji Zhao$^*$}
\affiliation{
  \institution{Key Laboratory of Intelligent Perception and Human-Machine Collaboration, ShanghaiTech University}
  \city{Shanghai}
  \country{China}}
\email{zhaodj@shanghaitech.edu.cn}

\author{Hu Fu}
\affiliation{
  \institution{Key Laboratory of Interdisciplinary Research of Computation and Economics, Shanghai University of Finance and Economics}
  \city{Shanghai}
  \country{China}}
\email{fuhu@mail.shufe.edu.cn}

\author{Zhihao Gavin Tang}
\affiliation{
  \institution{Key Laboratory of Interdisciplinary Research of Computation and Economics, Shanghai University of Finance and Economics}
  \city{Shanghai}
  \country{China}}
\email{tang.zhihao@mail.shufe.edu.cn}

\author{Pinyan Lu}
\affiliation{
  \institution{Key Laboratory of Interdisciplinary Research of Computation and Economics, Shanghai University of Finance and Economics}
  \city{Shanghai}
  \country{China}}
\email{lu.pinyan@mail.shufe.edu.cn}

\begin{abstract}
In cooperative games, we study how values created or costs incurred by a coalition are shared among the members within it, and the players may join the coalition in a online manner such as investors invest a startup.
Recently, Ge \textit{et al.}~\cite{ge2024incentives} proposed a new property called incentives for early arrival (I4EA) in such games, which says that the online allocation of values or costs should incentivize agents to join early in order to prevent mutual strategic waiting. Ideally, the allocation should also be fair, so that agents arriving in an order uniformly at random should expect to get/pay their Shapley values. Ge~\textit{et al.}~\cite{ge2024incentives} showed that not all monotone value functions admit such mechanisms in online value sharing games.

In this work, we show a sharp contrast in online \emph{cost} sharing games.  We construct a mechanism with all the properties mentioned above, for every monotone cost function.  To achieve this, we first solve 0-1 valued cost sharing games with a novel mechanism called \emph{Shapley-fair shuffle cost sharing mechanism (SFS-CS)}, and then extend SFS-CS to a family called \emph{generalized Shapley-fair shuffle cost sharing mechanisms (GSFS-CS)}.
The critical technique we invented here is a mapping from one arrival order to another order so that we can directly apply marginal cost allocation on the shuffled orders to satisfy the properties.   
Finally, we solve general valued cost functions, by decomposing them into 0-1 valued functions in an online fashion.
\end{abstract}


\newcommand{\BibTeX}{\rm B\kern-.05em{\sc i\kern-.025em b}\kern-.08em\TeX}

\begin{document}

%%% The following commands remove the headers in your paper. For final 
%%% papers, these will be inserted during the pagination process.

\pagestyle{fancy}
\fancyhead{}

%%% The next command prints the information defined in the preamble.

\maketitle 
\renewcommand{\thefootnote}{*}
\footnotetext{$^*$Corresponding author}
%%%%%%%%%%%%%%%%%%%%%%%%%%%%%%%%%%%%%%%%%%%%%%%%%%%%%%%%%%%%%%%%%%%%%%%%

\section{Introduction}

In a classical cost sharing game, a fixed group of players receive a service as a coalition, and the cost of the service is divided among the players. 
This game models many real-world applications such as electricity or water supply networks~\cite{bergantinos2010minimum,gomez2011merge,trudeau2017set}. 
They have been well studied and many solutions have been proposed to address different requirements~\cite{kar2002axiomatization,bergantinos2007optimistic,trudeau2012new}; e.g.\@ players should have incentives to stay together to share the cost~\cite{trudeau2020clique}, players' shares should be fair~\cite{shapley1953value}, or they should change consistently as members join or leave the coalition~\cite{immorlica2008limitations,thomson1995population}. 

Most traditional cost sharing games distribute the cost offline, i.e., the mechanism knows all the players and their cost function in advance. 
% This is largely the case, for example, when an electricity supply network is built in a suburb. 
However, in many applications, players typically do not all arrive at the same time; rather, they join sequentially. 
Upon the arrival of a new player, an irrevocable decision has to be made on the division of the current cost, without knowledge about players that arrive in the future. 

% For example, in a facility sharing game where a number of towns consider sharing the cost of a public facility such as a water distribution system, roads or communication networks, new towns may join after some initial coalition.
% The facility may need to be expanded to accommodate each new comer, so the total cost is a combinatorial function of the set of towns in the coalition, and at every point in time should be shared by the towns in the coalition.
% Typically, when deciding to join the coalition, a town needs to weigh its budget against the cost it is to share in the coalition.  
% It is important and natural to require that, once a town is in the coalition, its share should not increase in the future as more towns join.
% We term this condition as \emph{online individual rationality}.
% If a town can decrease its shared cost by unilaterally delaying its arrival, then the towns may strategically wait for each other, and end up in a deadlock.
For example, a shirt factory produces shirts for different customers with different brands and their orders do not arrive at the same time. However, the unit cost per shirt may decrease when more shirts are produced at one period (considering the cost of assembling a new product line). On the other hand, when new orders arrive, the factory cannot just wait forever to get more orders to start. Therefore, the existing customers have to bear the production costs.
If so, the unit cost of early arrival orders might be higher than the late arrival orders. This will incentivize customers to strategically wait each other.
From the factory point of view, we need to incentivize customers to place orders as soon as they need, so that the factory can fully utilize its capacity.
%To do so, the unit cost for early arrival orders should be decreasing when more orders arrive, which is what we are going to design in the work.

% To achieve this, it may involve the cost transfer between different customers, but it makes no sense to increase one’s payment after she has already provided the order.
%Hence, the factory always wants to gather more orders as soon as possible. However, it also suggests that customers may delay the order to catch a lower price. If a customer can decrease her payment by unilaterally delaying her order, then the customers may strategically wait for each other. Consequently, the garment factory cannot carry out production tasks. Recently, Ge \textit{et al.}~\cite{ge2024incentives} introduced the notion of \emph{incentives for early arrivals} (I4EA) precisely to prevent this from happening. 

Except for the \emph{incentives for early arrivals} (I4EA), the solution also needs to satisfy some other properties.
The cost share to a player should be non-increasing when more players are joining (\emph{online individual rationality} (OIR)). One intuition for OIR is that a player made his decision according to the cost allocation on his arrival, if the cost is going to increase with uncertainty in the future, the player may not even want to join the coalition in the first place. Finally, the cost share should be considered fair, otherwise, allocating all the costs to the last comer trivially satisfies I4EA and OIR.
As a measure of fairness, we require that, if the players arrive in a uniformly random order, each player's expected cost share is his Shapley value (\emph{Shapley fairness} (SF)).

When the cost function is supermodular, letting each newly joined player pay her marginal cost satisfies the three properties --- OIR, I4EA, and SF; but for more general cost functions, it is not clear that a mechanism exists with all three properties.
In fact, for \emph{value} sharing games, Ge \textit{et al.}~\cite{ge2024incentives} showed that not all value functions admit such online mechanisms.

In sharp contrast, in this work we show that in online cost sharing games, every monotone cost function admits an online mechanism that is I4EA, OIR, and SF.

Specifically, we propose a new mechanism called the \emph{Shapley-fair shuffle cost sharing mechanism (SFS-CS)} to solve \textbf{all} 0-1 cost sharing games. SFS-CS is further extended to a class called the \emph{generalized Shapley-fair shuffle cost sharing mechanism (GSFS-CS)}. The key technique used in GSFS-CS is shuffling the original arriving order into a new order which essentially helps us decide who should pay the cost. The shuffling is the key to guaranteeing Shapley fairness, which is also quite involved.

In summary, our main contributions advance the state of the art as follows:
\begin{itemize}
    %\item We model the setting of online cost sharing game and proposed the property of I4EA in this setting.
    \item For 0-1 valued monotone online cost sharing games, we propose a mechanism called \emph{Shapley-fair shuffle cost sharing mechanism (SFS-CS)} to satisfy OIR, I4EA, and SF. 
    \item Extending SFS-CS, we propose a class of mechanisms (GSFS-CS), which gives more flexibility on allocating the cost.
    \item For general online cost sharing games, we propose decomposing a general cost function into 0-1 valued cost functions to utilize GSFS-CS to satisfy the same properties. 
\end{itemize}

The remainder of the paper is organized as follows. Section \ref{related work} introduces the related work and Section \ref{model} defines the model and the desirable properties. We then propose SFS-CS to solve all the 0-1 monotone cost sharing games in Section \ref{sec:0-1} and extend it to a general class (GSFS-CS) in Section \ref{sec:gsfs}.
Finally, we extend the solution to general monotone cost sharing games in Section \ref{sec:general}. %Finally, we discuss future investigations.

\subsection{Other Related Work}
\label{related work}
Many studies have already considered online cost sharing.
For example, online multicast cost sharing, where players
arrive one by one and each connects to the root by greedily choosing a path minimizing its cost, was studied in~\cite{DBLP:conf/spaa/CharikarKMNS08}. It focuses on the price of anarchy of the Nash equilibrium. Besides, cost sharing games with private valuations in an online setting was studied in~\cite{DBLP:conf/ciac/BrennerS10} and they give a perfect characterization of both weakly
group-strategyproof and group-strategyproof online cost sharing mechanisms for this model. Furthermore, some studies considered the application of online cost sharing for demand-responsive transport systems and horizontal supply chains~\cite{DBLP:journals/tits/FuruhataDKODBCW15,zou2021online}. None of them has considered the incentives for early arrival. %for cost sharing.

There are many online mechanism design problems on other topics. For example, the online coalition formation game is studied in~\cite{FlamminiMichele2021On,DBLP:conf/esa/BullingerR23}, where the players with preference for different coalitions arrive online. The main objective is to effectively allocate players who join asynchronously into groups, with the overarching goal of maximizing the collective social welfare. Besides,~\cite{pavan2014dynamic,2017Dynamic,parkesonline} studied the auction mechanism design in dynamic settings, where players with private valuations of items will arrive or change over time. Additionally, mechanism design with diffusion incentives is also a new trend~\cite{2016Diffusion,2020Incentives,DBLP:journals/ai/LiHGZ22,DBLP:conf/atal/Zhao21}. This area of research focuses on incentivizing people to invite their neighbors in a social network to participate in an auction or a collaboration. We consider a different setting for cost sharing, where the players can decide when to arrive. To improve time efficiency, the goal is to guarantee that the players are benefited from early arrival.

Many studies also considered the cost sharing problem of minimizing the cost in the network~\cite{kar2002axiomatization,bergantinos2007fair,trudeau2012new}. They consider the minimum spanning tree problem and aim to allocate the cost of the spanning tree among the players. Furthermore, cost sharing under private valuation and connection control was studied in~\cite{DBLP:conf/atal/ZhangZGZ23}, where players have private valuation about the service and the connection control of the edge. They all consider the offline setting while we focus on the online cost sharing.

%%%%%%%%%%%%%%%%%%%%%%%%%%%%%%%%%%%%%%%%%%%%%%%%%%%%%%%%%%%%%%%%%%%%%%%%

\section{The Model}

\label{model}
An online cost sharing game is given by a triple $(N,c, \pi)$, where $N$ is a set of players, $c:2^N\rightarrow \mathbb{R}_{\geq 0}$ is a cost function, and $\pi \in \Pi(N)$ is a permutation of~$N$ ($\Pi(N)$ denotes the set of all permutations of~$N$).
Players arrive sequentially, in the order given by~$\pi$. For a coalition $S\subseteq N$, $c(S)$ is the cost caused by $S$. 
$c$ is \emph{normalized} if  $c(\emptyset) = 0$, and is \emph{monotone} if $\forall T\subseteq S\subseteq N$,  $c(S)\geq c(T)$. 
Throughout this work, we consider normalized and monotone cost functions. The following gives some formal notations and definitions. %the properties we need to satisfy.
\begin{itemize}
    \item Given an order $\pi$, we say $j \prec_\pi i$ if player $j$ arrives strictly earlier than player $i$. All these players and $i$ herself form a set $p(i,\pi)$, %denotes the set of players that arrive (weakly) earlier than $i$, including $i$,
    i.e., $p(i,\pi):= \{j \mid  j\prec_{\pi} i \}\cup\{i\}$.
    \item For a subset $S \subseteq N$, \emph{$c$ restricted to~$S$}, written as $c_{|S}$, is the set function $c_{|S}:2^S\rightarrow \mathbb{R}_{\geq 0}$ defined by $c_{|S}(T)=c(T), \forall T \subseteq S$; 
    \emph{$\pi$ restricted to $S$}, written as $\pi_{|S}$, is the permutation of $S$ defined by $i \prec_{\pi_{|S}} j$ if{f} $i \prec_\pi j$, for all $i, j \in S$.
    \item A subset $S \subseteq N$ is a \emph{prefix} of~$\pi$ if $S$ is the set of first $|S|$ players to arrive according to~$\pi$. We denote this as $S \sqsubseteq \pi$. 
    \item For an online cost sharing game $(N, c, \pi)$ and a prefix $S \sqsubseteq \pi$, the \emph{local cost sharing game on $S$} is the game $(S, c_{|S}, \pi_{|S})$.
\end{itemize}

\begin{definition}[Marginal Cost]
    Given a cost function $c$, player $i$'s \textbf{marginal cost (MC)} to a coalition $S\ni i$ is $$\mathtt{MC}(i,c,S) := c(S)-c(S\setminus \{i\}).$$
\end{definition}

\begin{definition}[Shapley Value] Given a cost function $c$, player $i$'s \textbf{Shapley value} is
    $$\mathtt{SV}_i(c) := \frac{1}{|N|!}\sum_{\pi \in \Pi(N)}\mathtt{MC}(i,c,p(i,\pi)) \:. $$
\end{definition}

For a monotone cost function, the marginal cost of any player in any coalition is non-negative. Hence, the Shapley value is also non-negative. 

For a cost sharing game to be online, we would like that, at any point of time, when the set of players that have arrived is~$S$, the cost caused by~$S$ should be allocated irrevocably among the players in $S$, and this allocation should be conducted using only information on $c$ and $\pi$ restricted to~$S$. 
The next definitions formalize this.

\begin{definition}% [Policies and Mechanisms]
\label{def:mechanism}
A \textbf{cost sharing policy} $\phi$ maps a cost sharing game $(N, c, \pi)$ to an $|N|$-tuple of allocations, so that $\phi_i(N, c, \pi) \geq 0$ is player~$i$'s cost share, and $\sum_i \phi_i(N, c, \pi) = c(N)$.

% \hufu{Not altogether happy with this definition.}
An \textbf{online cost sharing mechanism} 
 is given by a cost sharing policy $\phi$ for each stage of the game, so that after the arrival of each prefix $S \sqsubseteq \pi$, each player $i \in S$ gets a share of $\phi_i(S, c_{|S}, \pi_{|S})$.
\end{definition}

When more players join, we expect that the additional costs caused by them should not be born by the players arriving before them. 
In other words, we require each player's cost share to weakly decrease as more players arrive:
% Otherwise, they may quit when being asked to pay. We call this property online individual rationality. 

\begin{definition}
\label{def:OIR}
    An online cost sharing mechanism $\phi$ is \textbf{online individually rational} (OIR) for cost function~$c$ if, for any arrival order $\pi$ and any $T, S \sqsubseteq \pi$ with $T \subseteq S$, we have $\phi_i(T, c_{|T}, \pi_{|T}) \geq \phi_i(S, c_{|S}, \pi_{|S})$ for every player $i \in T$.
    
\end{definition}

% In order to incentivize players to join early, we need to maker sure that if a player chooses to arrive later than her actual arrival when the other players’ order of arrivals remains the same, her cost share should weakly increases. Formally, 
To prevent players from strategically waiting, we would like that a player's cost share should weakly increase if she chooses to unilaterally delay her arrival.
Formally,

 \begin{definition}
     \label{def: I4EA}
      An online cost sharing mechanism~$\phi$ is \textbf{incentivizing for early arrival} (I4EA) 
      % on a game $v$ 
      if for any player $i$, $\phi_i(N, c, \pi_1)\leq \phi_i(N, c, \pi_2)$ for all $\pi_1$ and $\pi_2$ such that $\pi_{1|N \setminus \{i\}} = \pi_{2|N \setminus \{i\}}$
      and $p(i,\pi_1) \subset p(i,\pi_2)$.
 \end{definition}

\begin{proposition}
    An online cost sharing mechanism~$\phi$ is I4EA if and only if for any player $i$, $\phi_i(N, c, \pi_1)\leq \phi_i(N, c, \pi_2)$ for any two arrival orders $\pi_1$ and $\pi_2$ with only player $i$ being delayed one position, %two adjacent players $i$ and $j$ flipped,
    i.e., $\pi_1=[\dots,i,j,\dots]$ and $\pi_2=[\dots,j,i,\dots]$.
\label{Pro:I4EA}
\end{proposition}

We would like an online cost sharing mechanism to be fair.  
To this end, we require a player’s expected cost share to be equal to her Shapley
value if the arrival order is uniformly at random.

\begin{definition}
     An online cost sharing mechanism~$\phi$ is \textbf{Shapley-fair} (SF) for a cost function~$c$ if 
     for each player $i\in N$, 
     \begin{align*}
     \frac 1 {|N|!} \sum_{\pi \in \Pi(N)} \phi_i(N, c, \pi)= \mathtt{SV}_i(c).
     \end{align*}
\end{definition}

In this work, we aim to design online cost sharing mechanisms that are OIR, I4EA, and SF.

\section{0-1 Valued Cost Functions}
\label{sec:0-1}

In this section, we study cost functions that take values only 0 or 1. Section~\ref{sec:shuffle_based} introduces the definition of the shuffle rule and the corresponding shuffle-based cost sharing mechanism. We then propose our our mechanism in Section~\ref{SFS-CS} and show its desirable properties in Section~\ref{sec:proofs}.

\subsection{Shuffle-Based Cost Sharing Mechanisms}
\label{sec:shuffle_based}

For 0-1 valued cost functions, for any arrival order~$\pi$, there is at most one player whose arrival makes the current coalition's cost jump from $0$ to~$1$.
We call this player the \emph{marginal player}.

\begin{definition}
         Given a $0$-$1$ valued cost function $c$ and an order $\pi$, player $i\in \pi$ is called the \textbf{marginal player} if $ \mathtt{MC}(i,c,p(i,\pi))=1$.
\end{definition}

As a first attempt, we consider a cost sharing game where $N=\{A,B,C\}$ and $c(S) = 1$ if and only if $A\in S$ or $\{B,C\}\subseteq S$. A possible allocation that satisfies I4EA, OIR, and SF is illustrated in Figure~\ref{fig:intui-ex}. For each order, we first shuffle the order and apply marginal cost allocation in the shuffled order. For example, we shuffle $[A,B,C]$ to $[B,C,A]$, in which $C$ is the marginal player and bears the cost.

\begin{figure}[htb]
    \centering
    \includegraphics[width=1.0\linewidth]{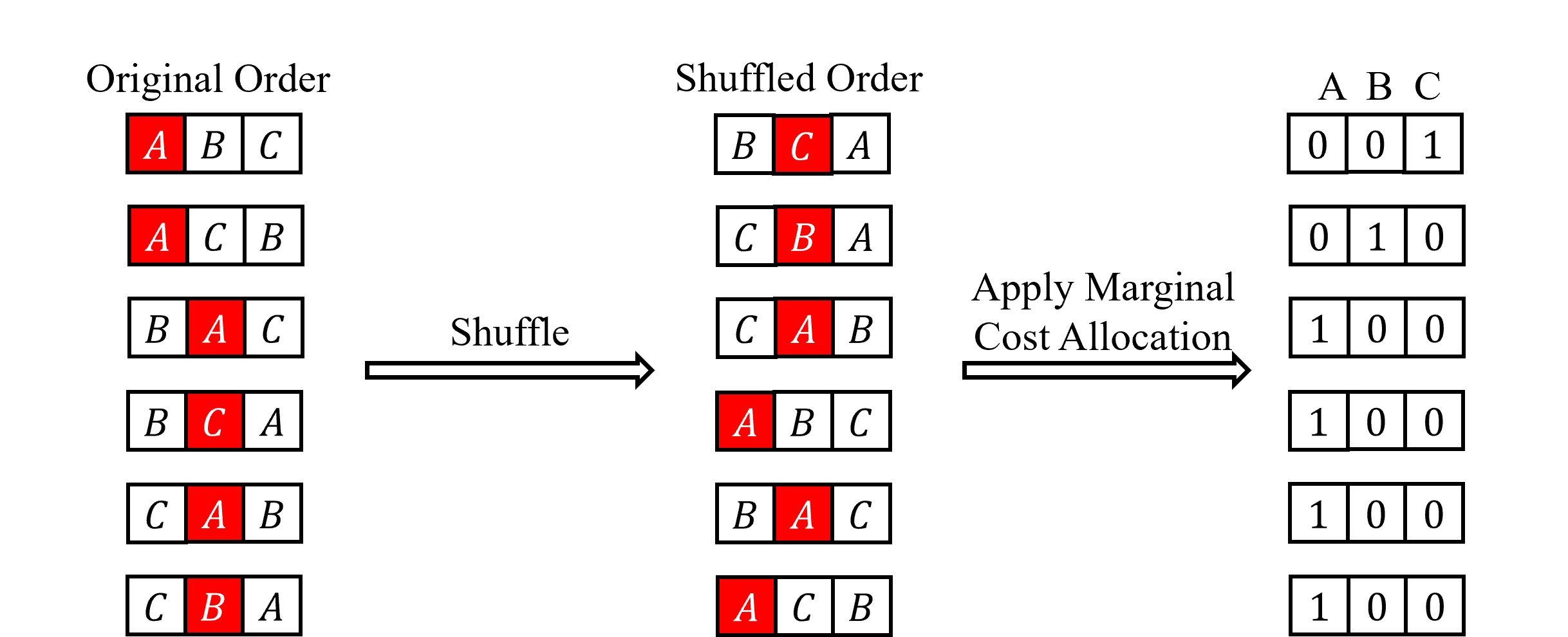}
    \caption{The player in red is the marginal player.}
    \label{fig:intui-ex}
\end{figure}

Adding a layer of sophistication to this idea, one may extend it to a family of mechanisms that retain its Shapley fairness: given a permutation~$\pi$, if we use a \emph{bijection} to map $\pi$ to another permutation~$\pi'$, and let each player in~$\pi$ pay her marginal cost in~$\pi'$, the resulting mechanism must be SF. The idea is natural enough, but for it to be useful, one has to be able to calculate the bijection in an online fashion; on top of that, the bijection must satisfy other properties for the resulting mechanism to be OIR and I4EA.
Our main result is that, somewhat surprisingly, a desirable mechanism based on such a bijection exists and can be computed efficiently.
We first define bijections usable in such mechanisms:

% To incentivize early arrivals, the cost needs to be transferred to the player who arrives later. For example, when a new player joins, if she is the new marginal player in the sequence without the current marginal player, the cost will be transferred to her, as illustrated in Figure~\ref{eg:I4EA-not-sf}. Unfortunately, such a method is verified to be not SF\footnote{Due to the limitation of the space, we propose the counter example in Figure~\ref{eg:I4EA-not-sf} without showing detail}.  An interesting observation of above mechanism is that the allocation in an order can correspond to the marginal cost of another order. Recall that SF requires every player’s expected cost share over all possible joining orders to be exactly her Shapley value. If we can improve above mechanism so that the result in each order corresponds to a different order, then SF can be satisfied.

% \begin{figure}[htb]
 %   \centering
    % \includegraphics[scale=0.5]{figure/eg-I4EA-not-SF.png}
    % \caption{$N=\{A,B,C,D\}$. For any $T$, we have $c(T)=1$ if and only if $\exists S\subseteq T, S\in \{(A,B,C),(C,D)\}$. Consider the arriving orders $\pi_1=[C,B,A,D], \pi_2=[C,D,B,A]$. For $\pi_1$, the cost is transferred from $A$ to $D$. For $\pi_2$, the cost is transferred from $D$ to $A$.}
    % \label{eg:I4EA-not-sf}
% \end{figure}

% Hence, a natural idea is to shuffle any original order $\pi$ to a image ordering $\pi'$  and  the cost share of each player is just her marginal cost in $\pi'$. We give the description formally.

\begin{definition}
          Given a player set $N$ and cost function $c$, a \textbf{shuffle rule} is a function $\mathtt{shuf}: \bigcup_{S \subseteq N}\Pi(S)  \rightarrow \bigcup_{S \subseteq N} \Pi(S)$, satisfying 
          \begin{enumerate}[(a)]
           \item for each $S \subseteq N$, $\shuf$ restricted to $\Pi(S)$ is a bijection from $\Pi(S)$ to $\Pi(S)$, and this bijection depends only on $c_{|S}$; and
          \item for any $T \subseteq N$, $\pi \in \Pi(T)$,  and $S \sqsubseteq \pi$, $\shuf(\pi_{|S}) = (\shuf(\pi))_{|S}$.
          In algebraic language, $\shuf(\cdot)$ and projection to any prefix commutes.
          \end{enumerate}
          
          % for any $S\subseteq N$, which only depends on $c_{|S}$, and maps a permutation $\pi\in \Pi(S)$ to another permutation~$\pi'\in \Pi(S)$; moreover, for any $T\sqsubseteq \pi$, $\mathtt{shuf}(\pi_{|T}) = \pi'_{|T}$.
\end{definition}
\begin{framed}

\noindent Given this definition of shuffle rule, the following is a well defined online mechanism:

\noindent A \textbf{shuffle-based cost sharing mechanism}~$\phi$
 is given by a shuffle rule $\mathtt{shuf}$, so that after the arrival of each prefix $S \sqsubseteq \pi$, each player $i \in S$ gets a share $\phi_i(S, c_{|S}, \pi_{|S})=\mathtt{MC}(i,c_{|S},p(i,\mathtt{shuf}(\pi_{|S})))$.
\end{framed}

Note that since $\shuf$ is a bijection on each $\Pi(S)$ for $S \subseteq N$, $\shuf$ is a bijection on $\bigcup_{S \subseteq N} \Pi(S)$.
Its inverse $\shuf^{-1}$ exists with $\shuf^{-1}(\shuf(\pi)) = \pi$ for any $\pi \in \bigcup_{S\subseteq N} \Pi(S)$.

%  Given a coalition $S$ and $\forall \pi_1, \pi_2 \in \Pi(S)$ such that $\pi_1 \neq \pi_2$, a shuffle rule $\mathtt{shuf}$ is \emph{injective} if $\mathtt{shuf}(\pi_1)\neq \mathtt{shuf}(\pi_2)$. $\mathtt{shuf}$ is \emph{surjective} if there exists  $\pi$ such that $\mathtt{shuf}(\pi)= \pi'$, $\forall \pi' \in \Pi(S)$. A shuffle rule $\mathtt{shuf}$ is \emph{bijective} if $\mathtt{shuf}$ is injective and surjective. Given an bijective shuffle rule $\mathtt{shuf}$, its inverse shuffle rule $\mathtt{shuf}^{-1}$ exists, with $\mathtt{shuf}^{-1}(\mathtt{shuf}(\pi))=\pi$, $\forall \pi \in \Pi(S)$.

\begin{proposition}
\label{pro:SF}
    A shuffle-based cost sharing mechanism 
 is SF.
 % if the shuffle rule is bijective.
\end{proposition}
\begin{proof} Write $\pi'=\mathtt{shuf}(\pi)$. Since $\mathtt{shuf}$ is bijective, we have
\begin{align*}
         \frac 1 {|N|!} \sum_{\pi \in \Pi(N)} \phi_i(N, c, \pi) = \frac 1 {|N|!} \sum_{\pi' \in \Pi(N)} \mathtt{MC}(i,c,p(i,\pi')) \\ =\frac 1 {|N|!} \sum_{\pi \in \Pi(N)} \mathtt{MC}(i,c,p(i,\pi)) =\mathtt{SV}_i(c).
\end{align*}

\end{proof}

% \begin{definition}[Consistent] For a 0-1 valued cost function~$c$, a shuffle rule $\mathtt{shuf}$ is \textbf{consistent} if for any $S \subseteq N$ and any $\pi \in \Pi(S)$, if player~$i$ is the marginal player in $\shuf(\pi)$, then $i$ is also the marginal player in $\shuf(\pi)_{|p(i,\pi)}$.
 % , where $\pi'=\mathtt{shuf}(\pi)$.
% \end{definition}

\begin{definition}[Group Size Monotone]
 For a 0-1 valued cost function~$c$, a shuffle rule $\mathtt{shuf}$ is \textbf{group size monotone} if for any $S \subseteq N$, any $\pi \in \Pi(S)$ with player $i$ being the marginal player in $\shuf(\pi)$, it is the case that for any $T\sqsubseteq \pi$ with $i\in T$, $i$ is also the marginal player in $\shuf(\pi_{|T})$.
 % , where $\pi'=\mathtt{shuf}(\pi)$.
\end{definition}

\begin{proposition}
\label{pro:cons}
    A shuffle-based cost sharing mechanism is OIR if its shuffle rule $\mathtt{shuf}$ is group size monotone.
\end{proposition}
\begin{proof} 
We only need to show that, for any arrival order $\pi \in \Pi(N)$ and any player~$i$, if $i$ is not the marginal player in the image ordering at a certain point, she never becomes the marginal player in the image ordering as more players join.
% Let $\pi'=\mathtt{shuf}(\pi)$.  
Suppose for prefix $T \sqsubseteq \pi$, $i$ is not the marginal player in~$\shuf(\pi_{|T})$.
Then as more players join, when the set of players is $S \supseteq T$, $i$ cannot be the marginal player in $\shuf(\pi_{|S})$ either, by the group size monotonicity of~$\shuf$.

 % For any player $i$, if she is not the marginal player in $\pi'_{|T}$ for $T\sqsubseteq \pi$ with $i\in T$, she is not the marginal player in $\pi'_{|T'}$ for any $T'\sqsubseteq \pi$ with $T\subseteq T'$, i.e., the cost share of player $i$ weakly decreases when more players join. Therefore, the mechanism is OIR.
\end{proof}

\begin{definition}[Flip Monotone]
For a 0-1 valued cost function, a shuffle rule $\mathtt{shuf}$ is \textbf{flip monotone} if for any two arrival orders with two adjacent players flipped, i.e., 
$\pi_1=[\dots,j,i,\dots]$ and $\pi_2=[\dots,i,j,\dots]$,
if player $i$ is the marginal player in $\mathtt{shuf}(\pi_2)$, then $i$ is the marginal player in $\mathtt{shuf}(\pi_1)$ as well.
\end{definition}

By Proposition~\ref{Pro:I4EA}, the following is immediate.
\begin{proposition}\label{prop:flip}
    A shuffle-based cost sharing mechanism 
 is I4EA if its shuffle rule $\mathtt{shuf}$ is flip monotone.
\end{proposition}
% \begin{proof} For the shuffle-based cost sharing mechanism, $\phi_i(N,c,\pi_1)=\mathtt{MC}(i,c,p(i, \mathtt{shuf}(\pi_1)))$ and $\phi_i(N,c,\pi_2)=\mathtt{MC}(i,c,p(i,\mathtt{shuf}(\pi_2)))$. If player $i$ is the marginal player in $\mathtt{shuf}(\pi_1)$, $i$ is the marginal player in $\mathtt{shuf}(\pi_2)$. Hence, $\phi_i(N,c,\pi_1)=\mathtt{MC}(i,c,p(i,\mathtt{shuf}(\pi_1)))\leq \mathtt{MC}(i,c,p(i,\mathtt{shuf}(\pi_2)))=\phi_i(N,c,\pi_2)$, i.e., the mechanism is I4EA.
% \end{proof}

\begin{theorem}
     A shuffle-based cost sharing mechanism 
 is SF, OIR, and I4EA if its shuffle rule $\mathtt{shuf}$ is group size monotone and flip monotone.
\end{theorem}

In Section~\ref{SFS-CS}, we construct a group size monotone and flip monotone shuffle rule $\mathtt{sfs}\text{-}\mathtt{shuf}$ and the corresponding \textit{Shapley-fair shuffle cost sharing mechanism} (SFS-CS).

\subsection{Shapley-fair Shuffle Cost Sharing Mechanism}
\label{SFS-CS}

We first make some observations and then illustrate some attempts at a desirable shuffle rule, before giving our final construction.

Given an arrival order~$\pi$, we refer to its image under $\shuf$ as the image ordering.  
When a new player~$i$ joins an existing coalition~$S$, the relative orderings of the players in~$S$ in the image ordering must remain unchanged by definition of shuffle rule; therefore, in the image ordering $\shuf(\pi_{|S \cup \{i\}})$, $i$ must be inserted into $\shuf(\pi_{|S})$.
This suggests that a construction of any shuffle rule can go by stages; at the arrival of each new player, one only needs to decide where to insert~$i$ in the image ordering.
We therefore maintain and grow an image ordering $\pi'$ as players join; when there is no danger of confusion, we treat $\pi'$ as a variable, and use it to denote $\pi'_{|S}$.  

Intuitively, both group size monotonicity and flip monotonicity suggest that if a late comer $i$ in~$\pi$ can be made a marginal player in the image ordering, $i$ should be made so.  
We will indeed follow this intuition in our construction; in fact, among all the positions in the image ordering that make $i$ the marginal player, we choose the earliest such position.  
The following example shows, however, that when a late comer cannot be made a marginal player in the image ordering, one needs to be very careful choosing her position, because the bijectivity of $\shuf$ may be at stake.

% In an online cost sharing game, players will arrive according to an order $\pi$. Initially, when the first player comes, a shuffle-based cost sharing mechanism can only maps the current order to itself. Let $\pi'$ be the image ordering updated by the corresponding shuffle rule after each player arrives. Before proposing our method, we first illustrate a possible direct idea to construct such a process. It gives an intuition of our method though it will not work.

\begin{example}
\label{ex:ABC}
    Consider the 0-1 valued cost sharing game with $N=\{A,B,C\}$, and $c(S) = 1$ if and only if $S \supseteq \{A, B\}$.
     For the arrival order $\pi_1=[A,B,C]$, when $A$ arrives, $\pi'_1=[A]$; then $B$ arrives; since $c(\{A,B\})=1$, to satisfy OIR, $\pi'_1$ must be $[A,B]$. 
     % Otherwise, $A$ will be the marginal player. 
     When $C$ arrives, she is not a marginal player in~$\pi'_1$ no matter where we insert her.  
     Does it make a difference where we insert her?
     % it seems that $C$ can be inserted at any position such as the start or the end of $\pi'_1$. Here
     \begin{itemize}
         \item If we insert $C$ to the start of $\pi'_1$, we get $\pi'_1=[C,A,B]$ (see Figure~\ref{fig:naive shuf eg}). 
         Consider another two orders, $\pi_2 = [A,C,B]$ and $\pi_3 = [C,A,B]$. 
         Since $\shuf$ is a bijection from $\Pi(\{A, C\})$ to itself, after the second player joins, one of the image orderings of $\pi_2'$ and $\pi_3'$ is $[C, A]$, and the other is $[A, C]$.
         $B$ is the last player in both $\pi_2$ and~$\pi_3$, and by group size monotonicity must be made the marginal player in $\pi_2'$ and~$\pi_3'$ when she joins.  
         The image ordering $[C, A]$ must now become $[C, A, B]$, contradicting $\shuf$ being a bijection.
         Therefore $\pi_1$ cannot be mapped to $[C, A, B]$.
         
         % We can observe that three different orders $\pi_1$, $\pi_2$ and $\pi_3$ can only have two different image orderings, i.e., it cannot be a bijection.  %For the order $\pi_2=[A,C,B]$, following the idea above, when $C$ arrives, if we we get $\pi'_2=[C,A,B]=\pi'_1$, i.e., it is not a bijection.
         \item By a similar argument, $C$ cannot be inserted at the end of $\pi_1'$.
         % If we insert $C$ at the end of $\pi'_1$, then similar to the above, it still cannot be a bijection.
     \end{itemize}
     % Hence, the only way to insert $C$ between $A$ and $B$ in $\pi_1'$.
     Therefore, $\pi'_1$ has to be $[A, C, B]$ when $C$ joins.
     
     % since it is more likely for the player in the future to be the marginal player of $\pi'$. Hence,
     
     % If we insert $C$ at the end of $\pi'_1$, 
     % % since it is more likely for the player in the future to be the marginal player of $\pi'$. Hence, 
     % we get $\pi'_1=[A,B,C]$. However, for the order $\pi_2=[A,C,B]$, following the idea above, we get $\pi'_2=[A,B,C]=\pi'_1$, i.e., it is not a bijection, either. 
    \label{eg:naive shuf}
\end{example}

\begin{figure}[htb]
    \centering
    \includegraphics[scale=0.5]{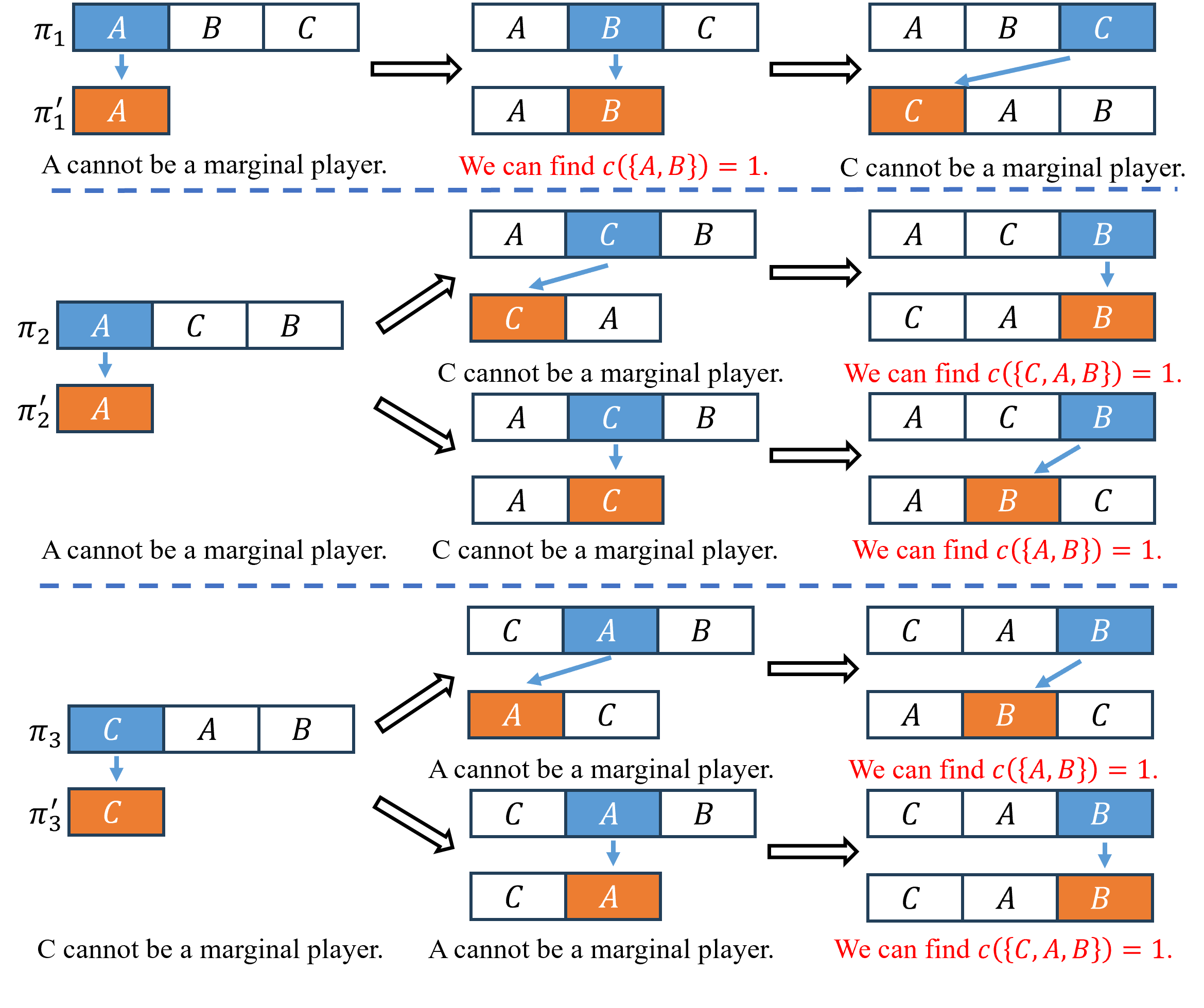}
    \caption{An attempt on constructing a $\mathtt{shuf}$ for cost sharing game mentioned in Example~\ref{eg:naive shuf}.}
    \label{fig:naive shuf eg}
    \Description{An attempt on constructing a shuffle rule.}
\end{figure}

It turns out that in Example~\ref{ex:ABC}, there is a shuffle rule mapping $\pi_1$ to $[A, C, B]$ that has all the desirable properties.
The construction of $\sfsshuf$ below shows that this is not a coincidence.
Specifically, when a newcomer~$i$ cannot be made the marginal player in the image ordering~$\pi'$, $\sfsshuf$ inserts $i$ in~$\pi'$ before \emph{$i$'s predecessor in~$\pi$}.  
(In Example~\ref{ex:ABC}, the predecessor of~$C$ in~$\pi$ is~$B$, hence $C$ is placed before~$B$ in~$\pi'$.)

We now formally construct $\sfsshuf$.
The iterative construction explicitly gives the image under $\sfsshuf$ of 
an order $\pi$ restricted to each prefix~$S \sqsubseteq \pi$. 
It should be clear that, in this mapping, $\sfsshuf$ uses only $\pi_{|S}$ and $c_{|S}$, and does produce an order on~$S$.
The fact that it is bijective is nontrivial, and is shown in Section~\ref{sec:proofs}.

\begin{framed}
% \begin{definition}
% \label{def:sfsshuf}
\noindent The shuffle rule $\sfsshuf$ maps any order $\pi$ to an image ordering~$\pi'$ given by the following iterative procedure:
\begin{itemize}
\item The image ordering $\pi'$ is initialized to be the first player in~$\pi$.
New players arriving according to~$\pi$ are iteratively inserted into~$\pi'$.
Let $i$ be the next player to arrive in~$\pi$.

\item \textbf{Case 1.} If $i$ can be inserted into~$\pi'$ so that $i$ becomes the marginal player in~$\pi'$, she is inserted into the earliest such position.
Formally, let $\mathcal P$ be the set $\{S \mid S \sqsubseteq \pi',  \MC(i, c, S \cup \{i\}) = 1\}$, then $\mathcal P \neq \emptyset$.  
Let $S^*$ be the member in~$\mathcal P$ with the smallest cardinality.
Update $\pi'$ so that $i$ is inserted after $S^*$.  
(Note that $S^*$ may be the empty set, in which case $i$ becomes the first player in~$\pi'$.)

\item \textbf{Case 2.} If there is no way to insert~$i$ into~$\pi'$ to make her the marginal player, update~$\pi'$ so that $i$ is inserted before her predecessor in~$\pi$.
(By predecessor we mean the player coming right before~$i$ in~$\pi$.)
\end{itemize}
% \end{definition}
\end{framed}

\begin{definition}
\label{def:SFS-CS}
The \emph{Shapley-fair shuffle cost sharing mechanism} (SFS-CS) is the shuffle-based cost sharing mechanism given by the shuffle rule~$\sfsshuf$.
\end{definition}

We show in Section~\ref{sec:proofs} the proofs of the main theorem.
\begin{theorem}
\label{The:SF-OIR-I4EA}
    For all 0-1 valued monotone cost sharing games, SFS-CS is SF, OIR, and I4EA.
\end{theorem}
% Obtaining the image ordering $\pi'=\mathtt{sfs}\text{-}\mathtt{shuf}(\pi)$, SFS-CS allocates the cost to the marginal player of $\pi'$. Here is a running example for SFS-CS.

Before the proof, we give an example of $\sfsshuf$.

\begin{example}
\label{eg:sfsshuf}
    Consider the 0-1 valued cost sharing game with $N=\{A,B,C,D,E,F,G\}$. For any $T$, $c(T)=1$ if and only if $\exists S\subseteq T, S\in \{\{A,C\},\{B,C\},\{B,D,E\},\{E,F\}\}$. For the arrival order $\pi=[A,B,C,D,E,F,G]$, when $A$ arrives, $\mathcal P=\emptyset$ and $\pi'=[A]$; then $B$ arrives and $\mathcal P=\emptyset$. 
    So $B$ is inserted before $A$ in $\pi'$, i.e., $\pi'=[B,A]$. When $C$ arrives, $\mathcal P=\{\{B\},\{B,A\}\}$ and $\pi'=[B,C,A]$. Then $D$ arrives and $\mathcal P=\emptyset$. So $D$ is inserted before $C$ in $\pi'$, i.e., $\pi'=[B,D,C,A]$. 
    The arrivals of $E, F$, and~$G$ are handled similarly.
    In the end, $E$ is the marginal player in the final image ordering, so her share is~$1$ and everyone else's is 0.
    % It is the similar for $E$ and $\pi'=[B,E,D,C,A]$. When $F$ arrives, $\mathcal P=\{\{B,E,D\}\}$ and $\pi'=[B,E,D,F,C,A]$. Finally $G$ arrives and she is inserted before $F$, i.e., $\pi'=[B,E,D,G,F,C,A]$. 
    % Therefore, $\phi_F(N,c,\pi)=1$ and otherwise $0$.
\end{example}
\begin{figure}[htbp]
\centering
    \includegraphics[scale=0.48]{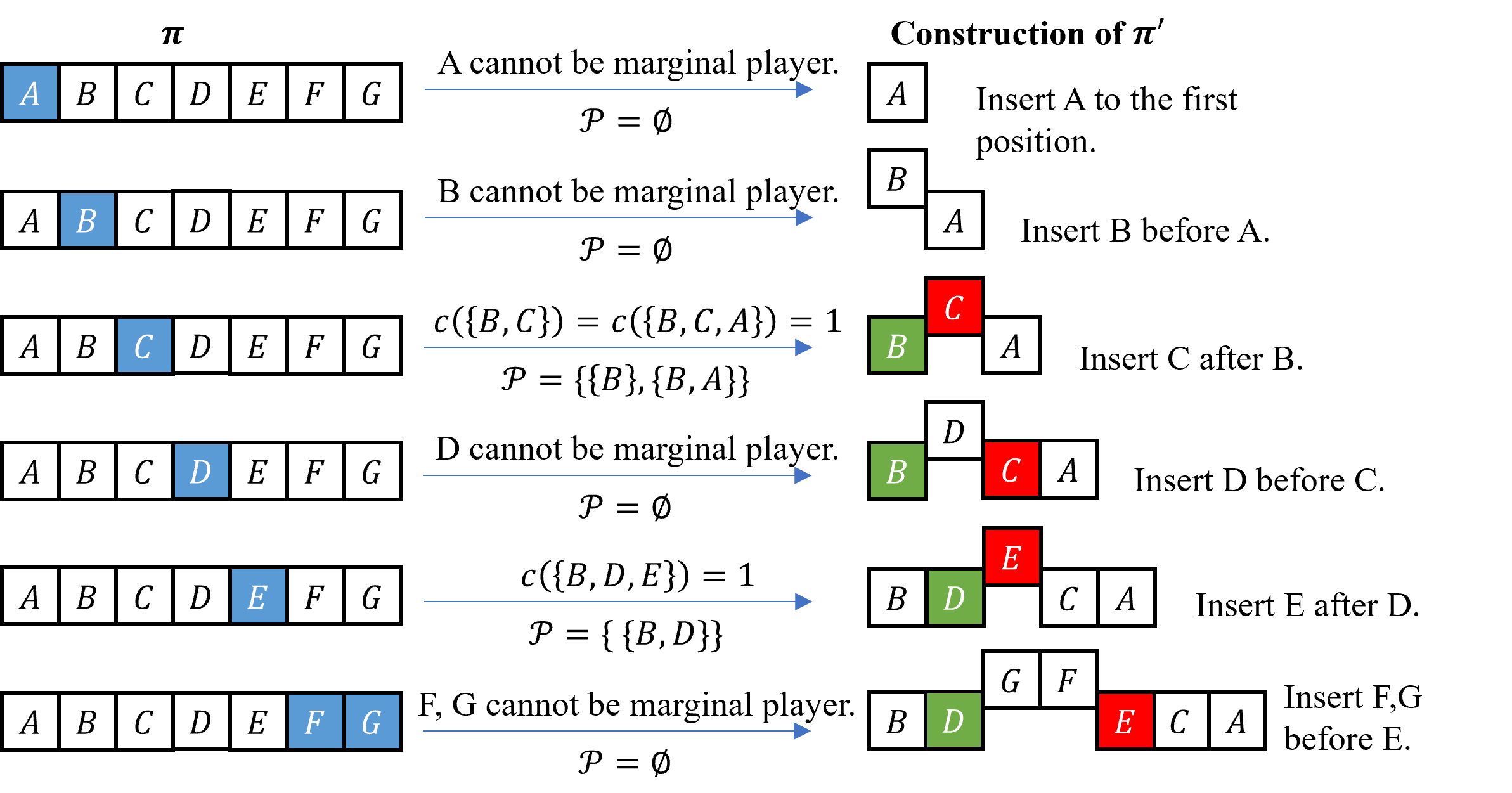}
    \caption{An example of $\sfsshuf$. The left side is the original order $\pi$ and for each joining player (colored by blue), the right side shows the construction process of the image ordering $\pi'$. Note that the players colored by red are the marginal players, and the players colored by green are the corresponding related players (see Definition~\ref{Def:related_player} in Section~\ref{sec:proofs}).}
    %{\color{red}}To be modified }
    \label{fig:SFS-CS-example}
    \Description{An example of sfsshuf.}
\end{figure}

% \begin{table}[]
% \caption{An example of $\sfsshuf$. The newly arriving player in each round is in red.}
% \centering
% \begin{tabular}{|c|c|c|c|}
% \hline
% $\pi$   & $\mathcal P$              & $\pi'$  & \text{ Marginal Player of $\pi'$ } \\ \hline
% A                & $\emptyset$                   & {\color{red}A}               & $\emptyset$                                \\ \hline
% AB               & $\emptyset$                   & {\color{red}B}A              & $\emptyset$                                \\ \hline
% ABC              & \text{ \{\{B\}\},\{B,A\}\} }  & B{\color{red}C}A             & C                                          \\ \hline
% ABCD             & $\emptyset$                   & B{\color{red}D}CA            & C                                          \\ \hline
% ABCDE            & $\emptyset$                   & B{\color{red}E}DCA           & C                                          \\ \hline
% ABCDEF           & \{\{B,E,D\}\}                & BED{\color{red}F}CA          & F                                          \\ \hline
% \text{ ABCDEFG } & $\emptyset$                   & \text{ BED{\color{red}G}FCA }& F                                          \\ \hline
% \end{tabular}
% \label{tb:SFS-CS-example}
% \end{table}

\subsection{Proofs of Properties of SFS-CS}
\label{sec:proofs}
In this section we prove Theorem~\ref{The:SF-OIR-I4EA}.

\subsubsection{OIR}
By Proposition~\ref{pro:cons}, it suffices to show that, if a player~$i$ is not the marginal player in the image ordering~$\pi'$ at a certain moment, then she does not become the marginal player in~$\pi'$ when a new player~$j$ joins.
Obviously, the only case that needs an argument is when $j$ is inserted into~$\pi'$, the cost of the new coalition is 1 and $j$ is not the marginal player in~$\pi'$; by definition of $\sfsshuf$, this happens only if $j$ cannot be made the marginal player no matter where she is inserted in~$\pi'$. 
For the sake of contradiction, suppose $i$ becomes the marginal player in~$\pi'$ after $j$ joins.  
This can happen only when $j$ is inserted before~$i$ in~$\pi'$.
Let $S$ be the prefix of~$\pi'$ up to~$i$ before $j$ joins. 
If $c(S) = 1$, then $c(S \setminus \{i\}) = 1$ (since $i$ was not the marginal player before $j$ joins), and $c(S \setminus \{i\} \cup \{j\}) \geq c(S \setminus \{i\}) = 1$, and so $i$ still cannot be the marginal player.
Therefore, $c(S) = 0$.  
If $i$ becomes the marginal player after $j$ joins, we must have $c(S \cup \{j\}) = 1$.  
But that means $j$ can be the marginal player if she is placed right after~$i$, a contradiction to the condition that $j$ cannot be a marginal player no matter where she is inserted.
This shows that $i$ cannot become the marginal player when a new player joins.

% and Lemma~\ref{Lem:gs-monotone}, SFS-CS is OIR.

% \begin{lemma}
    % $\sfsshuf$ is group size monotone.
    % \label{Lem:gs-monotone}
% \end{lemma}

% \begin{proof}
    % We prove it by contradiction. For any player $i$ who is not the marginal player in the image ordering $\pi'$ when she arrives, assume that when a new player $j$ joins, $i$ becomes the marginal player in $\pi'$. In this case, $j$ can be the marginal player in $\pi'$ if $j$ is inserted after $i$, which leads to a contradiction with the process of $\sfsshuf$. Hence, $\sfsshuf$ is group size monotone.
% \end{proof}

\subsubsection{SF}
By Proposition~\ref{pro:SF}, it suffices to show that $\mathtt{sfs}\text{-}\mathtt{shuf}$ is a bijection. 
We do this by constructing an inverse mapping for $\sfsshuf$.
Given an ordering $\pi' \in \Pi(S)$, we show that a unique $\pi \in \Pi(S)$ can be found such that $\sfsshuf(\pi) = \pi'$.
To this end, it suffices to show that we can uniquely identify the last player in~$\pi$, which allows us to iteratively reconstruct~$\pi$.

\begin{itemize}
\item If $c(S) = 0$, there cannot be a marginal player in~$\pi'$, in which case $\pi$ is simply the reverse of~$\pi'$.
\item If $c(S) = 1$, let $i$ be the marginal player in~$\pi'$.  
\begin{itemize}
\item If $c(\{i\}) = 1$, then 
\begin{enumerate}[(a)]
\item Either $i$ is the last comer in~$\pi$, in which case $i$ must be the first player in~$\pi'$;
% is the first player in~$\pi'$, in which case she is also the last player in~$\pi$;
\item Or $i$ is not the last comer in~$\pi$, but the marginal player in~$\pi'$ has not changed since $i$'s arrival in~$\pi$.  In this case, the players before~$i$ in~$\pi'$ are precisely the players arriving after~$i$ in~$\pi$, and their order in~$\pi'$ is precisely reverse to their order in~$\pi$.  
Hence the first player in~$\pi'$ is the last player in~$\pi$.
\end{enumerate}
\item If $c(\{i\}) \neq 1$, we consider the following two cases more carefully below:
 \begin{enumerate}[(a')]
     \item 
 Either $i$ is the last comer in~$\pi$;
 \item Or $i$ is not the last comer, but the marginal player in~$\pi'$ has not changed since $i$'s arrival in~$\pi$.
 \end{enumerate}
\end{itemize}
 \end{itemize}
In case (a) and (a'), by definition of~$\sfsshuf$, $i$ is in the earliest position in~$\pi'$ that makes her the marginal player.
In particular, in case~(a), $i$ has no predecessor in~$\pi'$, and in case~(a'), let $j$ be the predecessor of~$i$ in~$\pi'$, we must have $c(p(j, \pi')) = 0$, and for any $k \prec_{\pi'} j$, $c(p(k, \pi') \cup \{i\}) = 0$.

In case (b) and~(b'), let $T$ be the set of players arriving after~$i$ in~$\pi$, then players in~$T$ are inserted before~$i$ in~$\pi'$, in the order precisely reverse to $\pi_{|T}$.
In case~(b'), if we still let~$j$ be the predecessor of~$i$ in~$\pi'$ at the moment right after $i$'s arrival  in~$\pi$ (i.e., before the arrival of anyone in~$T$), then we have $c(p(j, \pi') \cup \{i\}) = 1$, and the person right after~$j$ in~$\pi'$ is the last one to arrive in~$\pi$.

To summarize, to distinguish case~(a') and~(b'), let $j \prec_{\pi'} i$ be the first player in~$\pi'$ with $c(p(j, \pi') \cup \{i\}) = 1$.
If $j$ is the predecessor of~$i$ in~$\pi'$, then $i$ is the last one to arrive in~$\pi$; otherwise, the player right after~$j$ is the last one to arrive in~$\pi$. Therefore, we can always uniquely identify the last player in~$\pi$, and iteratively reconstruct $\pi$ from~$\pi'$.
Example~\ref{eg:reconstruction} provides an illustration.

Two notions in the analysis of cases (a') and~(b') are used again in the next section.
We formally define them as follows.
% We give them names to avoid repetitions.
\begin{definition}
\label{Def:related_player}
    Given~$\pi'$ with the marginal player~$i$ and $c(\{i\}) = 0$, $i$'s predecessor in~$\pi'$ right after $i$'s insertion into~$\pi'$ is said to be $i$'s \textbf{related player}.
\end{definition}
In the proof above, $j$ is $i$'s related player and she can be identified as the first player in~$\pi'$ before~$i$ such that $c(p(j, \pi') \cup \{i\}) = 1$.

\begin{definition}
    The set of players arriving after the marginal player~$i$ in~$\pi$ is called the \textbf{late arriving set} in~$\pi'$, denoted as $\mathtt{LA}(\pi')$.
\end{definition}
In the proof above, the set~$T$ is the late arriving set. This set can be identified in~$\pi'$ as the set of players before $i$ if $c(\{i\}) = 1$, or, if $c(\{i\}) = 0$, as the set of players between $i$ and her related player~$j$. By this notation, we can simplify the reconstruction process by finding the players in late arriving set, reversing them, and putting them at the end of reconstruction order.
% \begin{itemize}
%   \item Given~$\pi'$ with the marginal player~$i$ and $c(\{i\}) = 0$, $i$'s predecessor in~$\pi'$ right after $i$'s insertion into~$\pi'$ (i.e., player~$j$ in the proof above) is said to be $i$'s \textbf{Related Player}.
% She can be identified as the first player in~$\pi'$ before~$i$ such that $c(p(j, \pi') \cup \{i\}) = 1$.
% \item The set of players arriving after the marginal player~$i$ in~$\pi$ (i.e., the set~$T$ in the proof above) is called the \textbf{Late Arriving Set} in~$\pi'$, which is denoted as $\mathtt{LA}(\pi')$.
%   This set can be identified in~$\pi'$ as the set of players before $i$ if $c(\{i\}) = 1$, or, if $c(\{i\}) = 0$, as the set of players between $i$ and her related player~$j$. By this notation, we can simplify the reconstruction process by finding the players in late arriving set, reversing them, and putting them at the end of reconstruction order.
%   \end{itemize}

\begin{figure}[htbp]
    \centering
    \includegraphics[scale=0.55]{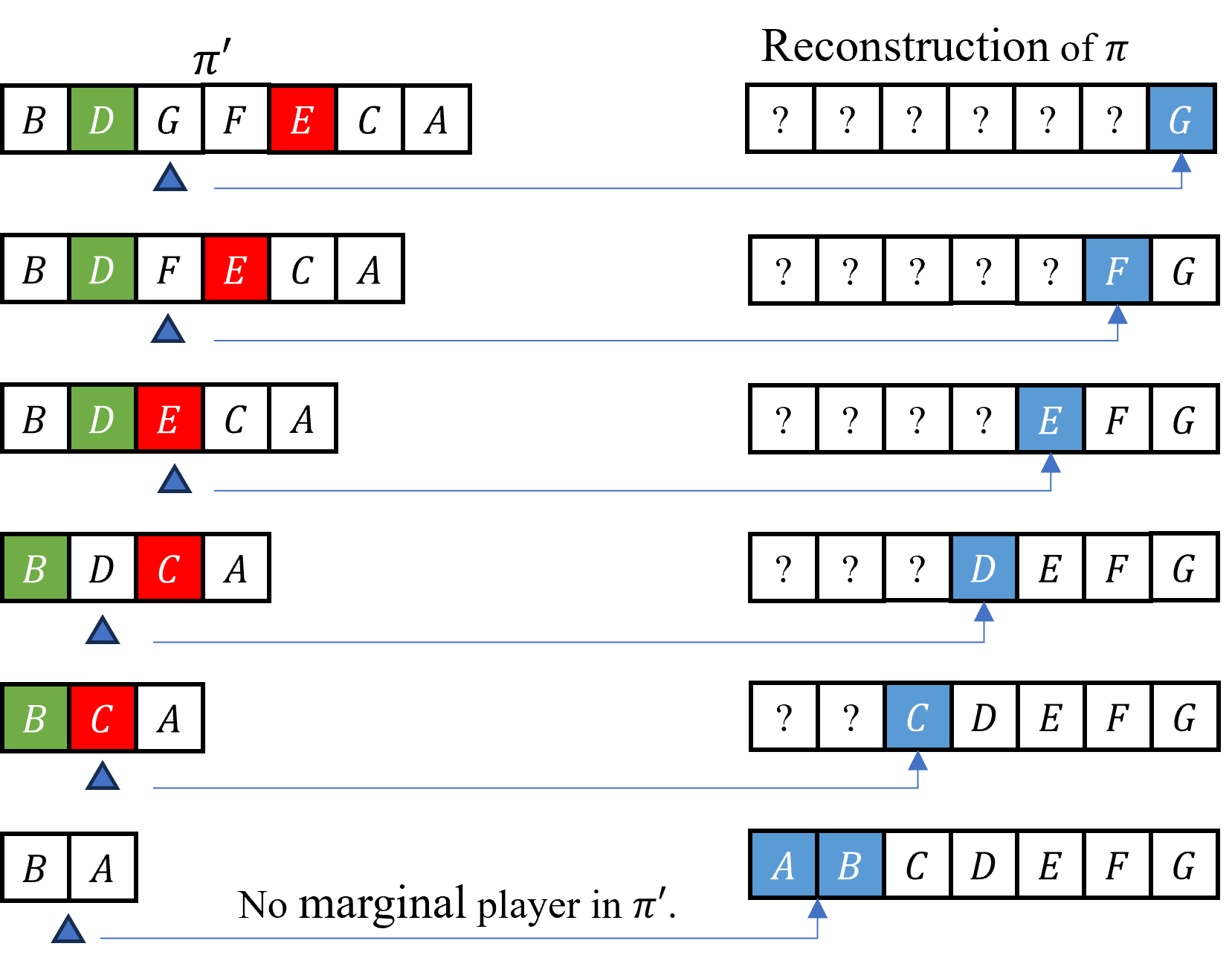}
    \caption{The reconstruction process of $\pi$ from $\pi'$. The left side marks the marginal players (red) and corresponding related players (green) through this process, and the right side shows the identified last players (blue) in each iteration. }
    \label{fig:RM-example}
    \Description{The reconstruction process of the order.}
\end{figure}

\begin{example}
\label{eg:reconstruction}
    Considering the game in Example~\ref{eg:sfsshuf}, with a final image ordering $\pi'=[B,D,G,F,E,C,A]$, we give the procedure to recover the original order (see Figure~\ref{fig:RM-example}). %we give the 0-1 valued monotonic cost sharing game with $N=\{A,B,C,D,E,F,G\}$. For any $T$, we have $c(T)=1$ if and only if $\exists S\subseteq T, S\in \{\{A,C\},\{B,C\},\{B,D,E\},\{E,F\}\}$. Consider the image ordering $\pi'=[B,D,G,F,E,C,A]$. 
    % Let $\pi''$ denote the reconstruction order in the process.
    Initially, %$\pi''=\emptyset$,
    the marginal player is $E$ and the related player is $D$, so that the player $G$ right after $D$ is the one who arrives last in $\pi$. Similarly, we can find the player $F$ is the one who arrives last among the remaining players. Next, the player right after the related player $D$ is player $E$ herself, so that $E$ is the last one arriving. Following steps are all similar and we can finally find the original order $\pi$. % and $\mathtt{LA}(\pi')=\{G,F\}$. The suborder $[G,F,E]$ is reversed and added to $\pi''$. $\pi'$ is updated to $[B,D,C,A]$. The marginal player is $C$ and the related player is $B$. Similarly, we get $\mathtt{LA}(\pi')=\{D\}$ and reverse $[D,C]$ and add it to $\pi''$. Finally, we reverse $[B,A]$ and add it to $\pi''$. It can be verified that the reconstruction order $\pi''=[A,B,C,D,E,F,G]=\pi$. % {\color{red}May need to be modified}
\end{example}

\subsubsection{I4EA}
Before we prove I4EA, we introduce two lemmas. Lemma~\ref{lem: I4EA1} states that if player $i$ is the marginal player of the image ordering when she joins, all the players after her in the original order are inserted before her in the image ordering. 

\begin{figure}[htb]
    \centering
    \includegraphics[scale=0.5]{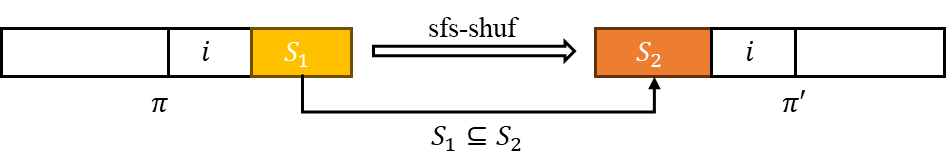}
    \caption{$\pi$ is the original order. $\pi'$ is the image ordering. $i$ is the marginal player of $\pi'$. Lemma~\ref{lem: I4EA1} states $S_1\subseteq S_2$.}
    \Description{A figure to explain Lemma 1.}
\end{figure}

\begin{lemma}
\label{lem: I4EA1}
    Given an order $\pi$ and image ordering $\pi'$ of $\sfsshuf$, suppose player $i$ is the marginal player of $\pi'_{|p(i,\pi)}$(i.e., the image ordering when $i$ just joins). $\forall j\in \pi$ satisfying $i\prec_{\pi} j$, we have $ j\prec_{\pi'}i$.
\end{lemma}

\begin{proof}
Consider player $i$'s next player $j$, there are two cases.
\begin{itemize}
    \item \textbf{Case 1.} If $j$ cannot be the marginal player, she is inserted before her predecessor $i$ of $\pi$, i.e., $j\prec_{\pi'}i$. 
    \item \textbf{Case 2.} If $j$ can be the marginal player, we prove the statement by contradiction. Assuming that $i\prec_{\pi'}j$, we have $c(p(j,\pi'_{|p(j,\pi)}))\geq c(p(i,\pi'_{|p(i,\pi)})) = 1$ and $c(p(j,\pi'_{|p(j,\pi)})\setminus \{j\})\geq c(p(i,\pi'_{|p(i,\pi)})) =1$. Hence, 
    \begin{align*}
        \mathtt{MC}(j,c,p(j,\pi'_{|p(j,\pi)})) &= c(p(j,\pi'_{|p(j,\pi)}))-c(p(j,\pi'_{|p(j,\pi)})\setminus \{j\})\\ &=c(p(i,\pi'_{|p(i,\pi)}))-c(p(i,\pi'_{|p(i,\pi)}))\\ &=1-1 =0,
    \end{align*}
    which leads to a contradiction.
\end{itemize}
For $j$'s next player $k$, (i) if $j$ cannot be the marginal player, there are two cases.
\begin{itemize}
    \item \textbf{Case 1.} If $k$ cannot be the marginal player, she is inserted before her predecessor $j$ of $\pi$, i.e., $k\prec_{\pi'}j\prec_{\pi'}i$. 
    \item \textbf{Case 2.} If $k$ can be the marginal player, the analysis is the same as the Case 2 above so that $k\prec_{\pi'}i$.
\end{itemize}
(ii) If $j$ is the marginal player, the analysis of $k$ is the same as the analysis of $j$ above so that $k\prec_{\pi'}j\prec_{\pi'}i$.

For the following players, we can get the conclusion recursively.
\end{proof}

We can use Figure~\ref{fig:lem: I4EA2} to visually illustrate the insights of Lemma~\ref{lem: I4EA2}.
\begin{figure}[htb]
    \centering
    \includegraphics[scale=0.45]{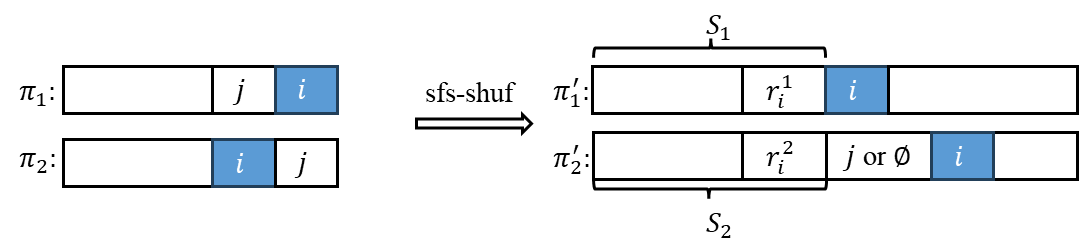}
    \caption{$\pi_1, \pi_2$ are the original orders, where only last two players flip. $\pi'_1$ and $\pi'_2$ are the corresponding image orderings. $i$ is the marginal player of $\pi'_1$ and $\pi'_2$. $r^1_i$ and $r^2_i$ are the corresponding related players. Lemma~\ref{lem: I4EA2} states $S_1 \setminus\{j\} \subseteq S_2$.}
    \label{fig:lem: I4EA2}
    \Description{A figure to explain Lemma 2.}
\end{figure}

\begin{lemma}
\label{lem: I4EA2}
    Given two orders $\pi_1=[\dots,j,i],\text{and } \pi_2=[\dots,i,j]$, where only adjacent $i$ and $j$ exchange. $\pi'_1$ and $\pi'_2$ are the corresponding image orderings of $\sfsshuf$. If player $i$ is the marginal player of both $\pi'_1$ and $\pi'_2$, and let $r^1_i$ and $r^2_i$ be the related players of $\pi'_1$ and $\pi'_2$, then we have $p(r^{1}_i,\pi'_1)\setminus\{j\}\subseteq p(r^{2}_i,\pi'_2)$.
\end{lemma}

\begin{proof}
For player $j$, there are two cases.
\begin{itemize}
    \item \textbf{Case 1.} $j \prec_{\pi'_1} r^2_i$. For $\pi'_1$, when $i$ is inserted after $r^2_i$, $i$ is the marginal player. Hence, $i$ cannot be inserted in the later position and we can get $r^1_i\prec_{\pi'_1} r^2_i$ or $r^1_i=r^2_i$. Therefore, $p(r^{1}_i,\pi'_1) \setminus \{j\} \subseteq p(r^{2}_i,\pi'_2)$.
    \item \textbf{Case 2.} $r^2_i \prec_{\pi'_1} j$. When $i$ joins for $\pi_1$, $\sfsshuf$ must insert $i$ after $r^2_i$ such that $i$ is the marginal player. Hence, $r^1_i=r^2_i$ and thus $p(r^{1}_i,\pi'_1) \setminus \{j\} \subseteq p(r^{2}_i,\pi'_2)$. 
\end{itemize}
\end{proof}

Now we give the proof of I4EA with Lemma~\ref{lem: I4EA1}, Lemma~\ref{lem: I4EA2}, and Proposition~\ref{prop:flip}.

\begin{proof}[Proof of I4EA]
 Given  $\pi_1=[\dots,j,i,\dots], \pi_2=[\dots,i,j,\dots]$, where only adjacent $i$ and $j$ exchange. 
 
 \noindent (i) If $c(\{i\})=1$, $i$ is inserted at the head of the image orderings when she joins. According to Lemma~\ref{lem: I4EA1}, $p(i,\pi'_1)\subseteq p(i,\pi'_2)$. Hence, $c(p(i,\pi'_1)) = c(p(i,\pi'_2)) = 1$ and $c(p(i,\pi'_1)\setminus \{i\}) \leq c(p(i,\pi'_2)\setminus \{i\})$, which derives
 \begin{align*}
     \mathtt{MC}(i,c,p(i,\pi'_1)) &=c(p(i,\pi'_1))-c(p(i,\pi'_1)\setminus \{i\}) \\& \geq c(p(i,\pi'_2))-c(p(i,\pi'_2)\setminus \{i\}) =\mathtt{MC}(i,c,p(i,\pi'_2)),
 \end{align*}
i.e., if $i$ is the marginal player in $\pi'_2$, she is the marginal player in $\pi'_1$. 

% $\phi_i(N,c,\pi_1)\geq \phi_i(N,c,\pi_2)$.

\noindent (ii) If $c(\{i\})=0$, there are three cases for player $i$ in $\pi_1$.

\begin{itemize}
    \item \textbf{Case 1.} $i$ is the marginal player for $\pi'_1$. In this case, $\mathtt{sfs}\text{-}\mathtt{shuf}$ is flip monotone.
    \item \textbf{Case 2.} $i$ is the marginal player for $\pi'_{1|p(i,\pi_1)}$ but is not for $\pi'_1$. Let $S_1=\{k\mid k\in \pi_1, i\prec_{\pi_1} k\}$. For $k\in S_1$, we have $k\prec_{\pi'_1}i$ (according to Lemma~\ref{lem: I4EA1}). We then prove flip monotone in this case by contradiction. Assume that $\mathtt{sfs}\text{-}\mathtt{shuf}$ is not flip monotone, i.e., $i$ is the marginal player for $\pi_2'$ but is not the marginal player for $\pi_1'$. Similarly, let $S_2=\{k \mid k\in \pi_2, i\prec_{\pi_2} k\}$. For $k\in S_2$, we have $k\prec_{\pi'_2}i$ (according to Lemma~\ref{lem: I4EA1}). We can observe that $S_1 \cup \{j\} = S_2$. Let $r^1_i$ and $r^2_i$ denote the related players of $\pi'_{1|p(i,\pi_1)}$ and $\pi'_{2|p(j,\pi_2)}$, we have $p(r^{1}_i,\pi'_{1|p(i,\pi_1)}) \setminus \{j\} \subseteq p(r^{2}_i,\pi'_{2|p(j,\pi_2)})$ (according to Lemma~\ref{lem: I4EA2}). Hence, 
    \begin{align*}
        p(i,\pi_1')\setminus\{i\} &=S_1\cup p(r^{1}_i,\pi'_{1|p(i,\pi_1)})\\& \subseteq S_2\cup p(r^{2}_i,\pi'_{2|p(j,\pi_2)})\\ &= p(i,\pi_2')\setminus\{i\}.
    \end{align*}
    Since $i$ is not the marginal player for $\pi_1'$, $c(p(i,\pi_2')\setminus\{i\})\geq c(p(i,\pi_1')\setminus\{i\})=1$, i.e., $i$ is not the marginal player for $\pi_2'$, which leads to a contradiction.
    \item \textbf{Case 3.} $i$ is not the marginal player for $\pi'_{1|p(i,\pi_1)}$. There are two cases for player $j$.
    \begin{itemize}
        \item \textbf{Case 3.1.} $j$ is not the marginal player of $\pi'_{1|p(j,\pi_1)}$. Hence, for $\pi_2$, $i$ still cannot be the marginal player of $\pi'_{2|p(i,\pi_2)}$ without $j$'s participation, i.e., $\mathtt{sfs}\text{-}\mathtt{shuf}$ is flip monotone.
        \item \textbf{Case 3.2.} $j$ is the marginal player for $\pi'_{1|p(j,\pi_1)}$. We then prove flip monotone in this case by contradiction. Assume that $\mathtt{sfs}\text{-}\mathtt{shuf}$ is not flip monotone, i.e., $i$ is the marginal player for $\pi_2'$ but is not the marginal player for $\pi_1'$. Let $r_i^2$ denote the corresponding related player. Since $i$ is not the marginal player for $\pi'_{1|p(i,\pi_1)}$, we can get $j\prec_{\pi_1'} r_i^2$; otherwise, $i$ can be inserted after $r^2_i$ such that $i$ is the marginal player for $\pi'_{1|p(i,\pi_1)}$. Therefore, for $\pi_2$, $j$ can be inserted in the same location of $\pi_2'$ with that of $\pi_1'$ such that $j$ is the marginal player for $\pi'_{2|p(j,\pi_2)}$, i.e., $i$ is not the marginal player for $\pi_2'$, which leads to a contradiction.
    \end{itemize}
\end{itemize}

Taking all above together, we can conclude that SFS-CS is I4EA.

\end{proof}

\section{A Class of Shuffle-based Cost Sharing Mechanisms}
\label{sec:gsfs}
In this section, we propose a class of shuffle-based cost sharing mechanisms based on SFS-CS. Notice that the key point of SF is to ensure that the shuffle rule is a bijection. Recalling the process of $\mathtt{sfs}\text{-}\mathtt{shuf}$ in Example~\ref{eg: SFS and GSFS} below, we can see more possibilities. 

\begin{example}
\label{eg: SFS and GSFS}
Consider the 0-1 valued cost sharing game with $N=\{A,B,C,D,E\}$. For any $T$, we have $c(T)=1$ if and only if $\exists S\subseteq T, S\in \{\{A\},\{B,D\},\{C,E\}\}$. For the arrival order $\pi=[A,B,C,D,E]$, when player $C$ arrives, $\sfsshuf$ inserts her before her predecessor $B$. Actually, in this step, we can observe that inserting player $C$ between player $B$ and $A$ is also available. Intuitively, it will not affect how we find the late arriving set from the image ordering $\pi'$ and we can recover the order of the late arriving set by recording the insertion positions; hence, its properties will not be hurt. Interestingly, on the other hand, it may change the player who finally bears the cost; as in this example, player $E$ cannot be the marginal player in the new image orderings when she arrives. %Moreover, when $D$ joins, there are three possible positions. $\pi'_1$ is the image ordering of SFS-CS and $\pi'_2$ is one of the possible image orderings. We can observe that the marginal players of the two image orderings are different.
\end{example}
\begin{figure}[htbp]
\centering
    \includegraphics[scale=0.48]{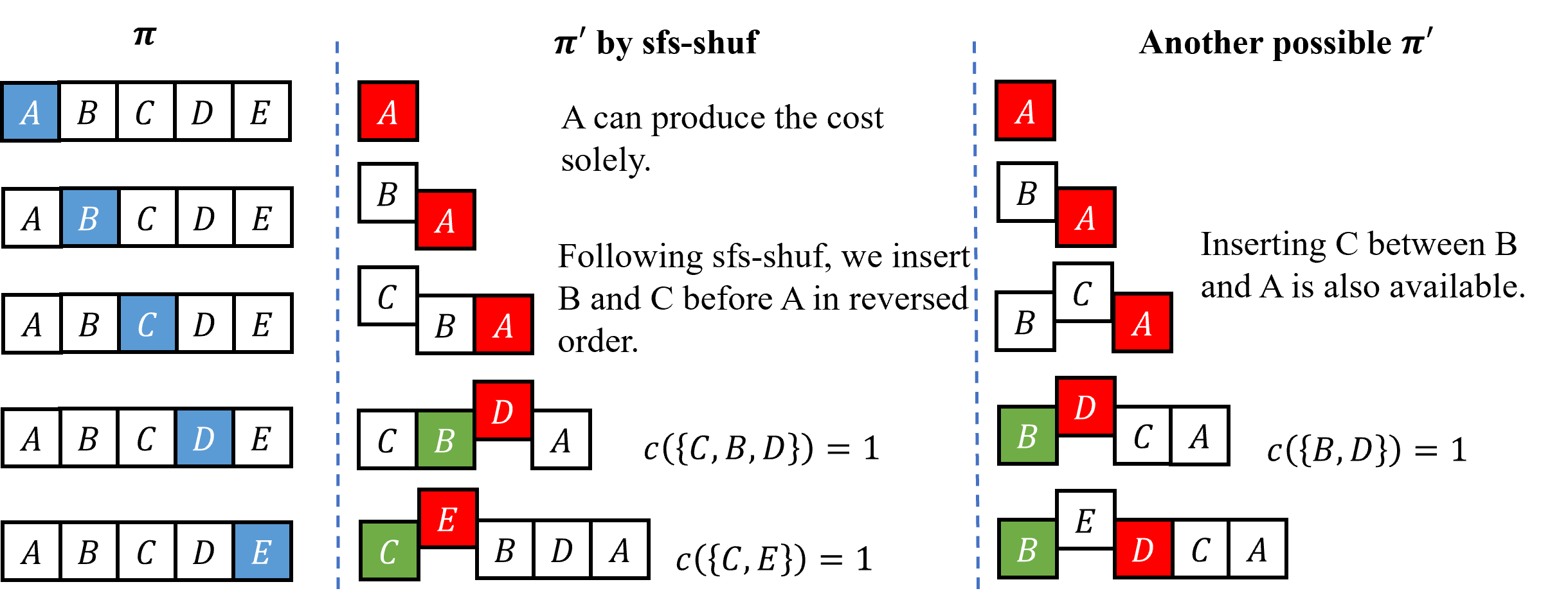}
    \caption{An example of $\sfsshuf$ and another possible mapping. The left side is the original order $\pi$ and for each joining player (colored by blue), the middle and right side show the image orderings given by these methods. The players colored by red are the marginal players, and the players colored by green are the corresponding related players.}
    %{\color{red}}To be modified }
    \label{fig:gSFS}
    \Description{A example of another mapping.}
\end{figure}

% By the given intuitions, the key is  how to rearrange the sequence of the following arrived players who cannot be new marginal players. $\sfsshuf$ produces a reverse order of their original order, while actually, any bijective rearrangement can yield a valid shuffle rule. Formally, we use the following coordinate functions for each player $i$ to decide those rearrangements after $i$ becomes the marginal player.

% By the given intuitions, the key is after a marginal player, how to rearrange the sequence of the following arrived players who cannot be new marginal players. $\sfsshuf$ produces a reverse order of their original order, while actually, any bijective rearrangement can yield a valid shuffle rule. Formally, we use the following coordinate functions for each player $i$ to decide those rearrangements after $i$ becomes the marginal player.

By the given intuitions, the key is how to decide the rearrangements of the following arrived players who cannot be new marginal players. For the case 2 of $\sfsshuf$, it produces a reverse order of their original order, while actually, any bijective rearrangement can yield a valid shuffle rule. Formally, we define the following coordinate functions to decide those rearrangements as follows.

% When player $i$ joins and she cannot be the marginal player of $\pi'$, she is inserted in $\pi'$ before her predecessor of $\pi$. In fact, the insertion position may not be unique. If there is no marginal player in the image ordering, $i$ can be inserted in any position.

\begin{definition}[Coordinate function]
     Given a player set $N$ and a player $i\in N$, a \textbf{coordinate function} of $i$ is a bijection function $\mathtt{cd}_i: \bigcup_{S \subseteq N\setminus\{i\}}\Pi(S)  \rightarrow \bigcup_{S \subseteq N\setminus\{i\}} \Pi(S)$, and for any $T \subseteq N\setminus\{i\}$, $\pi \in \Pi(T)$,  and $S \sqsubseteq \pi$, $\mathtt{cd}_i(\pi_{|S}) = (\mathtt{cd}_i(\pi))_{|S}$.
\end{definition}
When the marginal player $i$ of $\pi'$ exists, we use $\mathtt{cd}_i$ to decide the rearrangements of those players who cannot be new marginal players after $i$. Especially, When there is no marginal player, we use $\mathtt{cd}_0$ to represent the rearrangements before the first marginal player appears.

%to denote how we insert the player $i$ in this case, i.e., the relative position of the players. Otherwise, let $j$ denote the marginal player, $i$ can be inserted in any position such that $i\in \mathtt{LA}(\pi')$. We use $\mathtt{map}_i$ denote the relative position of $i$ in $\mathtt{LA}(\pi')$ for this case. 
% \begin{table}[]
% \centering
% \caption{$\pi'_1$ is the image ordering of $\sfsshuf$ and $\pi'_2$ is one of the possible image orderings. The players in red denote the marginal players in the image orderings.} 
% \begin{tabular}{|c|c|c|}
% \hline
% $\pi$ & $\pi'_1$ & $\pi'_2$ \\ \hline
% A      & A        & A       \\ \hline
% AB     & BA       & BA       \\ \hline
% ABC    & CBA      & BCA      \\ \hline
% ABCD   & DCBA     & BCDA     \\ \hline
% ABCDE  & DCB{\color{red}E}A    & B{\color{red}E}CDA    \\ \hline
% \text{ ABCDEF } & \text{ D{\color{red}F}CBEA }   & \text{ BF{\color{red}E}CDA }   \\ \hline
% \end{tabular}
% \label{tb:GSFS-CSs-example}
% \end{table}
Now we can choose different coordinate functions and get a class of shuffle rules described as follows, which extends the case 2 of $\sfsshuf$.

\begin{framed}
\label{def:gsfsshuf}
\noindent The $\gsfsshuf$ maps any order $\pi$ to an image ordering~$\pi'$ given by the following iterative procedure:
\begin{itemize}
\item The image ordering $\pi'$ is initialized to be the first player in~$\pi$.
% New players arriving according to~$\pi$ are iteratively inserted into~$\pi'$.
Let $i$ be the next player to arrive in~$\pi$.

\item \textbf{Case 1.} If $i$ can be inserted into~$\pi'$ so that $i$ becomes the marginal player in~$\pi'$, she is inserted into the earliest such position. 
% \item Insert $i$ to the earlist position that $i$ can be marginal player if possible.
% Formally, let $\mathcal P$ be the set $\{S \mid S \sqsubseteq \pi',  \MC(i, c, S \cup \{i\}) = 1\}$, then $\mathcal P \neq \emptyset$.  
% Let $S^*$ be the member in~$\mathcal P$ with the smallest cardinality.
% Update $\pi'$ so that $i$ is inserted after $S^*$.  
% (Note that $S^*$ may be the empty set, in which case $i$ becomes the first player in~$\pi'$.)

\item \textbf{Case 2.} If there is no way to insert~$i$ into~$\pi'$ to make her the marginal player, there are two cases for $\pi'$.
\begin{itemize}
    \item If there is no marginal player, update $\pi'$ so that $\pi'=\mathtt{cd}_0(\pi_{|p(i,\pi)})$.
    \item If the marginal player $j$ exists, update $\pi'$ so that $i\in \mathtt{LA}(\pi')$ and $\pi'_{|\mathtt{LA}(\pi')}=\mathtt{cd}_j(\pi_{|\mathtt{LA}(\pi')})$.
\end{itemize}
\end{itemize}
\end{framed}

% \begin{framed}
%  \noindent\textbf{A Class of Shuffle Rules: $\mathtt{gsfs}\text{-}\mathtt{shuf}$ }
 
%  \noindent\rule{\textwidth}{0.5pt}
 
%  \noindent\textbf{Input}: 0-1 monotonic cost sharing game $(N,c,\pi)$, map set $ \{\mathtt{map}_0\}\cup\{\mathtt{map}_i\}_{i\in N}$.
 
%  \noindent\rule{\textwidth}{0.5pt}
% $\pi'=[]$\\
% For player $i\in \pi$: 
%  \begin{enumerate}
%  \item $S=\{j|j\in \pi'\}$.
%  \item $\Pi'=\{\pi''\mid \pi''\in\Pi(\{S, i\}), \pi''\setminus\{i\}=\pi', \mathtt{MC}(i,c,p(i,\pi''))=1\}$, i.e., $\Pi'$ contains all the orders resulted by inserting $i$ in $\pi'$ such that $i$ is the marginal player.
%  \item If $\Pi'$ is empty, find the marginal player $j$ of $\pi'$. Insert player $i$ in $\pi'$ such that $i\in \mathtt{LA}(\pi')$ and $\pi'_{|\mathtt{LA}(\pi')}=\mathtt{map}_j(\pi_{|\mathtt{LA}(\pi')})$. If there is no marginal player, insert player $i$ in $\pi'$ such that $\pi'=\mathtt{map}_0(\pi)$.
%  \item Otherwise, update $\pi'=\underset{\pi''\in\Pi'}{\arg\min} |p(i,\pi'')|$.
%  \end{enumerate}
 
%  \noindent\rule{\textwidth}{0.5pt}
 
%  \noindent\textbf{Output}: the image ordering $\pi'$.
% \end{framed}
\begin{definition}
\label{def:GSFS-CS}
The \emph{generalized Shapley-fair shuffle cost sharing mechanism} (GSFS-CS) is the shuffle-based cost sharing mechanism given by the shuffle rule~$\gsfsshuf$.
\end{definition}

Now we can also get the reconstruction of $\mathtt{gsfs}\text{-}\mathtt{shuf}$. Compared to the reconstruction of $\mathtt{sfs}\text{-}\mathtt{shuf}$, the only difference is that %converses the suborder of the late arriving set and then updates $\pi'$. For $\mathtt{gsfs}\text{-}\mathtt{shuf}^{-1}$,
in each iteration, we use the inverse of $\mathtt{cd}_i$ to reposition the corresponding late arriving set. %, which is determined by the marginal player $i$.

% \begin{framed}
% \noindent The inverse $\gsfsshuf^{-1}$ maps any image ordering $\pi'$ to an original order~$\pi''$ given by the following iterative procedure:
% \begin{itemize}
% \item $\pi''$ is initialized to be empty;

% \item If the marginal player $i$ of $\pi'$ exists, find the late arriving set $\mathtt{LA}(\pi')$. Let $\Tilde{\pi}=\mathtt{cd}^{-1}_i(\pi'_{|\mathtt{LA}(\pi')})$, and update $\pi''=[i, \Tilde{\pi}, \pi'']$. Remove player $i$ and $\mathtt{LA}(\pi')$ in $\pi'$, and repeat this step.

% \item If the marginal player $i$ of $\pi'$ does not exist, let $\Tilde{\pi}=\mathtt{cd}^{-1}_0(\pi')$. Update $\pi''=[\Tilde{\pi}, \pi'']$.
% \end{itemize}
% \end{framed}

% \begin{framed}
%  \noindent\textbf{The Inverse of $\mathtt{gsfs}\text{-}\mathtt{shuf}$: $\mathtt{gsfs}\text{-}\mathtt{shuf}^{-1}$}
 
%  \noindent\rule{\textwidth}{0.5pt}
 
%  \noindent\textbf{Input}: 0-1 valued monotonic cost function $c$ and the image ordering $\pi'$.
 
%  \noindent\rule{\textwidth}{0.5pt}
% $\pi''=[], \pi_1=\pi'$.\\
% While the marginal player $i$ of $\pi_1$ exists: 
%  \begin{enumerate}
%  \item Find the the late arriving set $\mathtt{LA}(\pi')$.
%  \item Let $\pi_2=\mathtt{map}^{-1}_i(\pi_{1|\mathtt{LA}(\pi')})$.
%  \item Update $\pi''=[i, \pi_2, \pi]$. $S_1=\mathtt{LA}(\pi')\cup\{i\}, S_2=\{j|j\in\pi_1\}$. Update $\pi_1=\pi_{1\mid (S_2\setminus S_1)}$.
%  \end{enumerate}
% Let $\pi_2=\mathtt{map}^{-1}_0(\pi_1)$. Update $\pi''=[\pi_2, \pi'']$. \\
%  \noindent\rule{\textwidth}{0.5pt}
 
%  \noindent\textbf{Output}: the reconstructed order $\pi''$.
% \end{framed}

Since $\mathtt{cd}_0$ and $\mathtt{cd}_i$ are bijective, it can be easily verified that $\mathtt{gsfs}\text{-}\mathtt{shuf}$ is still bijective. Moreover, we can observe that $\mathtt{gsfs}\text{-}\mathtt{shuf}$ is still group size monotone, and I4EA, Lemma~\ref{lem: I4EA1} and Lemma~\ref{lem: I4EA2} still hold and I4EA can be obtained in the same approach; it can be directly validated by noticing that those proofs only require players who cannot be marginal player are inserted between the marginal player and the related player. Hence, GSFS-CS is SF, OIR, and I4EA.
\begin{theorem}
    For all 0-1 valued monotone cost sharing games, GSFS-CS is OIR, I4EA, and SF.
\end{theorem}

% {\color{red}\begin{proof}[Sketch.]
% \textbf{OIR:} The class of shuffle rule is still consistent and hence GSFS-CS is OIR.

% \textbf{I4EA:} The idea is similar with Theorem~\ref{The:I4EA}. Lemma~\ref{lem: I4EA1} and Lemma~\ref{lem: I4EA2} still hold for any bijective shuffle rule $\mathtt{shuf}_i$. When player $i$ joins earlier, her cost share weakly decreases.

% \textbf{SF:} With similar analysis, we can prove $\mathtt{gsfs}\text{-}\mathtt{shuf}$ is bijective and hence GSFS-CS is SF.
% \qed
% \end{proof}}

\section{Extension to General Cost Functions}\label{sec:general}
\balance
% For a general monotone cost function $c$, if we can find a decomposition that gives a linear combination of 0-1 cost sharing games such that $c=\sum_{k}\mu_k g_k$, where $\{g_k\}$ are the 0-1 monotone cost sharing components and $\{\mu_k\}$ are the non-negative coefficients, then by the additivity of Shapley value, the sum of the solution by GSFS-CS on each component game $g_k$ is also a Shapley fair solution on $c$. Furthermore, if the decomposition is consistent in global and local games, such a decomposition way can be operated online. Then, we can get a valid online cost sharing mechanism for general monotone cost sharing games that is OIR, I4EA and SF. Actually, \cite{ge2024incentives} has proposed the greedy monotone decomposition (GM) that meets the requirements, and we can construct a corresponding extended GSFS-CS based on GM.

So far we have proposed a class of mechanisms satisfying all our requirements on 0-1 valued monotone cost sharing games. In this section, we show how GSFS-CS can be applied to general valued monotone cost sharing games. According to \cite{ge2024incentives}, an online value sharing mechanism on 0-1 valued monotone game can be extended to general valued setting while maintaining the properties by the following two steps: (1) decompose a general valued monotone function into positive-weighted 0-1 valued monotone components in an online fashion, and (2) run the mechanism simultaneously on each component game and determine each player's share as the weighted sum of her shares from those games.
We show that GSFS-CS can be extended to general valued setting through greedy-monotone decomposition (GM), which is an online decomposition algorithm proposed in [10] and meets the requirements. Furthermore, this extended mechanism can satisfy SF, OIR, and I4EA on any monotone cost sharing games.

\begin{lemma}
    Given $(N,c,\pi)$, the output $D(c)=\{(g_k,\mu_k)\}$ of GM-decomposition is a set of pairs where $c=\sum_{k}\mu_k g_k$. Note that $\{g_k\}$ are 0-1 valued monotone functions and $\{\mu_k\}$ are non-negative coefficients. 
\end{lemma}

% a methodological approach involves their decomposition into a weighted aggregation of 0-1 monotone cost functions. Subsequently, GSFS-CS is executed concurrently on each individual 0-1 cost component. The cost allocation for each participant is derived from the weighted aggregation of their respective cost shares obtained from the decomposed 0-1 cost sharing sub-games.

% For the decomposition algorithm,~\cite{ge2024incentives} proposed the greedy monotone decomposition (GM), which gives a non-negative linear combination of a monotone game as $c=\sum_{k}\mu_kg_k$, where $\{g_k\}$ are the 0-1 cost sharing components and $\{\mu_k\}$ are the coefficients. 
% They prove the GM-Decomposition is characterized by a positive linear combination of set functions, monotonicity of component functions, and consistent decomposition in global and local games, which is significant for extending GSFS-CS.

% Now we propose the extended GSFS-CS based on GM. The mechanism firstly does GM-decomposition on input cost function $c$. Then it calculates the cost share in each 0-1 cost monotone games by GSFS-CS and accumulates them with coefficients to be the cost share in $c$. The properties of GSFS-CS are maintained through this process.
\begin{definition}
    The \textbf{extended generalized Shapley-fair cost sharing mechanisms} (eGSFS-CS) is defined by $$\Bar{\phi}_i(S,c_{|S},\pi_{|S}) = \sum_{(g_k,\mu_k)\in D(c_{|S})}\mu_k\phi_i^{\mathtt{GSFS-CS}}(S,g_k,\pi_{|S})$$
    where $\phi_i^{\mathtt{GSFS-CS}}$ is the cost sharing policy of GSFS-CS. %and $D(\cdot)$ is the GM-decomposition.
\end{definition}

\begin{theorem}
    eGSFS-CS is SF, OIR, and I4EA.
\end{theorem}

\begin{proof}
    \textbf{SF}: Since the Shapley value satisfies \emph{additivity}, we have
    \begin{align*}
        \Bar{\phi}_i(N,c,\pi) = \sum_{(g_k,\mu_k)\in D(c)}\mu_k\phi_i^{\mathtt{GSFS-CS}}(N,g_k,\pi) \\ = \sum_{(g_k,\mu_k)\in D(c)}\mu_k \mathtt{SV}_i(g_k) =\mathtt{SV}_i(c).
    \end{align*}

    \textbf{OIR}: Given $\pi$, for any $T,S\sqsubseteq \pi$ with $T\subseteq S$, we have
    \begin{align*}
            \Bar{\phi_i}(S,c_{|S},\pi_{|S}) = \sum_{(g_k,\mu_k)\in D(c_{|S})}\mu_k \phi^{\mathtt{GSFS-CS}}_i(S,g_k,\pi_{|S}) \\ \leq \sum_{(g_k,\mu_k)\in D(c_{|T})}\mu_k \phi^{\mathtt{GSFS-CS}}_i(T,g_{k|T},\pi_{|T}) =  \Bar{\phi_i}(T,c_{|T},\pi_{|T}).   
    \end{align*}

    % \begin{align*}
    %     \Bar{\phi_i}(S,c_{|S},\pi_{|S})&= \sum_{(g_k,\mu_k)\in D(c_{|S})}\mu_k \phi^{\mathtt{GSFS-CS}}_i(S,g_k,\pi_{|S}) \\
    %                        &\leq \sum_{(g_k,\mu_k)\in D(c_{|T})}\mu_k \phi^{\mathtt{GSFS-CS}}_i(T,g_{k|T},\pi_{|T})  =  \Bar{\phi_i}(T,c_{|T},\pi_{|T}).                          
    % \end{align*}    
    
    \textbf{I4EA}:
    For $\pi_1 = [\dots,i,j,\dots]$ and $\pi_2 = [\dots,j,i,\dots]$, where only adjacent $i$ and $j$ exchange, we have 
    \begin{align*}
        \Bar{\phi_i}(N,c,\pi_1) = \sum_{(g_k,\mu_k)\in D(c)}\mu_k \cdot \phi^{\mathtt{GSFS-CS}}_i(N,g_k,\pi_1) \\  \leq \sum_{(g_k,\mu_k)\in D(c)}\mu_k \cdot \phi^{\mathtt{GSFS-CS}}_i(N,g_k,\pi_2) = \Bar{\phi_i}(N,c,\pi_2).
    \end{align*}
\end{proof}

\section{Discussion}
For the online 0-1 cost sharing problem, we formalize the shuffle-based cost sharing mechanism and propose SFS-CS mechanism, which satisfies OIR, SF, and I4EA. Based on this, we propose a class of mechanisms (GSFS-CS) satisfying the desirable properties. We also extend it to online general cost sharing problems by decomposing a general cost function into 0-1 cost functions.

Note that for value sharing games, not all value functions admit such online mechanisms. We discuss here on why cost sharing behaves so differently from value sharing. For 0-1 value sharing games, once the value of 1 is created, the mechanism decides whether to keep the value for the currently arrived player or transfer the value to the players who arrived earlier, which is a one-time allocation. However, for 0-1 cost sharing, when the cost is created, we can only allocate it to the currently arrived player (due to OIR). When more players arrive in the future, the mechanism decides whether to reallocate the cost to the new players, so the cost allocation keeps involving. Hence, the incentives for cost sharing and value sharing are fundamentally different. Figure \ref{Comparison} illustrates the difference of value sharing and cost sharing.

\begin{figure}[htb]
    \centering
    \includegraphics[width=8cm]{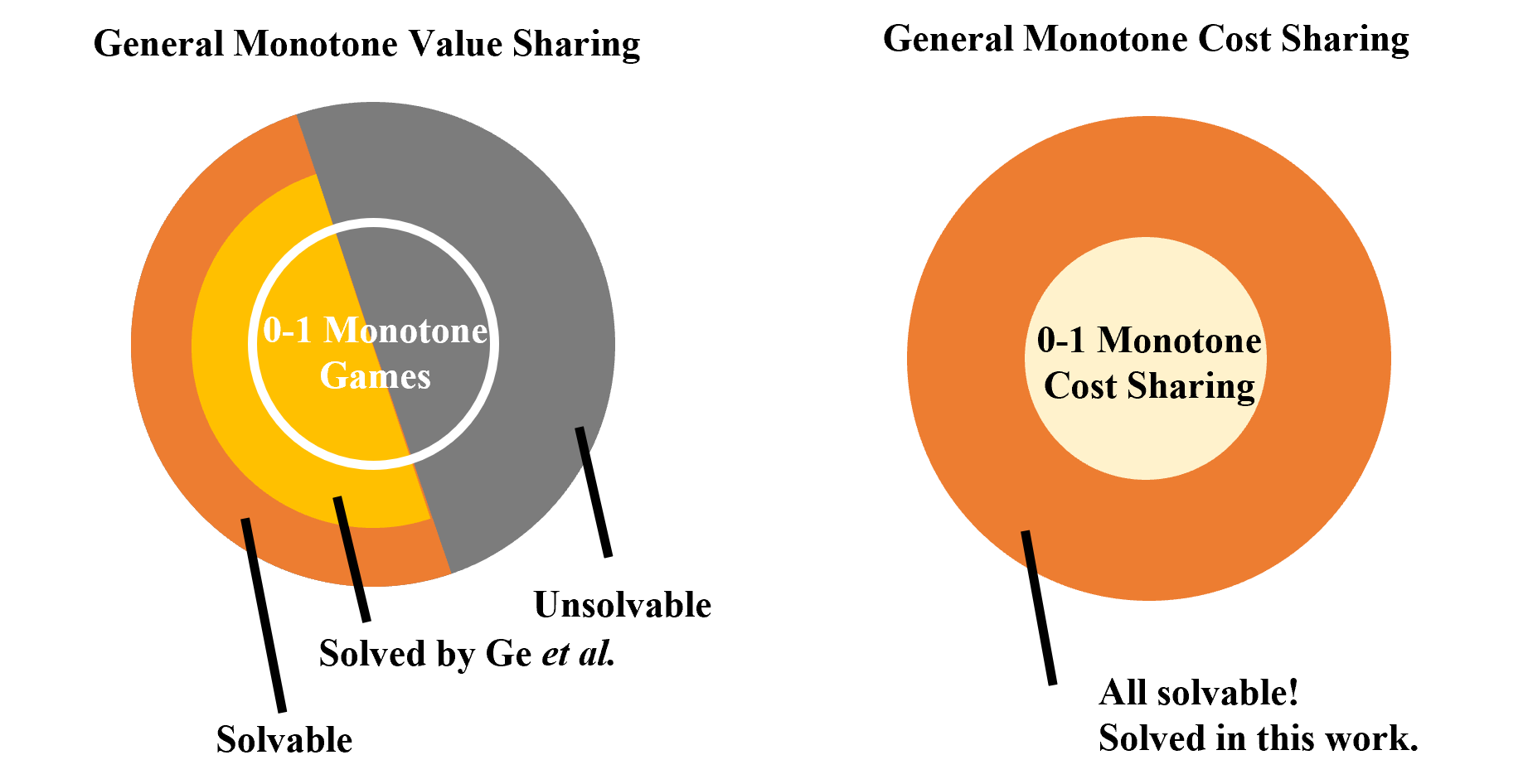}
    \caption{I4EA problems in value sharing and cost sharing. }
    \label{Comparison}
\end{figure}

There are several future directions worth investigation.
For 0-1 valued monotone cost sharing games, one may consider characterizing the whole set of mechanisms satisfying the desirable properties. For general monotone cost sharing games, since the time complexity of decomposition is exponential, one may consider designing polynomial time mechanisms.

\bibliographystyle{ACM-Reference-Format} 
\bibliography{sample}

%%%%%%%%%%%%%%%%%%%%%%%%%%%%%%%%%%%%%%%%%%%%%%%%%%%%%%%%%%%%%%%%%%%%%%%%

\end{document}